%% file: main.tex
\documentclass{sig-alternate}

\toappear{}

\usepackage{graphicx}
\usepackage{balance}  

\usepackage{booktabs}
\usepackage{amsmath}
\usepackage{verbatim}
\usepackage{algorithmicx}
\usepackage[plain]{algorithm}
\usepackage[noend]{algpseudocode}
\usepackage{pgfplots}
\usepackage{pgfplotstable}
\usepackage[scriptsize]{subfigure}

\pgfplotsset{every linear axis/.append style={width=4.5cm, height=3.5cm}}
\pgfplotsset{every axis plot/.append style={mark size=0.15cm},compat=newest}
\usetikzlibrary{patterns,shapes,arrows,external}
\usetikzlibrary{external}
\usepgfplotslibrary{groupplots}

\usepackage{xparse}
\let\oldState\State
\RenewDocumentCommand{\State}{o}{
  \IfValueTF{#1}{\makeatletter\setcounter{ALG@line}{#1}\addtocounter{ALG@line}{-1}\makeatother}{}%
  \oldState\ignorespaces%
}

\makeatletter
\newcommand\resetstackedplots{
\makeatletter
\pgfplots@stacked@isfirstplottrue
\makeatother
\addplot [forget plot,draw=none] coordinates{(1,0) (2,0) (3,0)};
}
\makeatother

\title{Approximation Schemes for \\Many-Objective Query Optimization}

\numberofauthors{1} 
\author{
\alignauthor
Immanuel Trummer and Christoph Koch\\
       \affaddr{\'Ecole Polytechnique F\'ed\'erale de Lausanne}\\
       \email{\{firstname\}.\{lastname\}@epfl.ch}
}

\begin{document}

\maketitle

\newcommand*{\codeF}{\fontfamily{\sfdefault}\selectfont}

\newlength{\abovecaptionskip}

\newcommand{\reMetricOne}{Buffer Space}
\newcommand{\reMetricTwo}{Time}
\newcommand{\planCost}{Plan Cost}
\newcommand{\optimalCost}{Optimal Cost}
\newcommand{\paretoFrontier}{Pareto Frontier}
\newcommand{\approxParetoFrontier}{Approximate \\Pareto Frontier}
\newcommand{\weights}{Weights}
\newcommand{\bounds}{Bounds}
\newcommand{\dominated}{Dominated Area}
\newcommand{\approxDominated}{Approximately\\Dominated Area}

\newcommand{\AG}{\textsf{EXA}}
\newcommand{\AWM}{\textsf{RTA}}
\newcommand{\ABM}{\textsf{IRA}}

\newcommand{\PW}{weighted MOQO}
\newcommand{\PB}{bounded-weighted MOQO}

\newtheorem{lemma}{Lemma}
\newtheorem{theorem}{Theorem}
\newtheorem{corollary}{Corollary}
\newtheorem{example}{Example}
\newtheorem{assumption}{Assumption}
\newtheorem{observation}{Observation}
\newtheorem{remark}{Remark}

\newdef{definition}{Definition}
\newdef{scenario}{Scenario}


\newcommand{\firstConfIndex}{6}
\newcommand{\secondConfIndex}{5}
\newcommand{\thirdConfIndex}{4}
\newcommand{\fourthConfIndex}{3}

\newcommand{\oneObjWeighted}{v23_1Obj_0Bou}
\newcommand{\threeObjWeighted}{v23_3Obj_0Bou}
\newcommand{\sixObjWeighted}{v23_6Obj_0Bou}
\newcommand{\nineObjWeighted}{v23_9Obj_0Bou}

\newcommand{\threeObjBounded}{v29_9Obj_3Bou}
\newcommand{\sixObjBounded}{v29_9Obj_6Bou}
\newcommand{\nineObjBounded}{v29_9Obj_9Bou}

\newcommand{\memoryOpt}{memoryOpt}

\newlength{\benchWidth}
\setlength{\benchWidth}{\columnwidth}
\newcommand*{\horScale}{0.975}

\pgfplotscreateplotcyclelist{benchmarkCycle}{%
        {black,pattern=horizontal lines,mark=none, /pgf/bar shift=-3pt},%
        {black,fill=black!60!white,mark=none, /pgf/bar shift=-1pt},%
        {black,fill=black!30!white,mark=none, /pgf/bar shift=1pt},%
        {black,fill=white,mark=none, /pgf/bar shift=3pt},
        {black,fill=black, /pgf/bar shift=-3pt}}

\pgfplotscreateplotcyclelist{benchmarkExtractCycle}{%
        {black,fill=black,mark=none, /pgf/bar shift=-2pt},%
        {black,fill=black!60!white,mark=none, /pgf/bar shift=0pt},%
        {black,fill=black!30!white,mark=none, /pgf/bar shift=2pt}}

\begin{abstract}
The goal of multi-objective query optimization (MOQO) is to find query plans that realize a good compromise between conflicting objectives such as minimizing execution time and minimizing monetary fees in a Cloud scenario. A previously proposed exhaustive MOQO algorithm needs hours to optimize even simple TPC-H queries. This is why we propose several approximation schemes for MOQO that generate guaranteed near-optimal plans in seconds where exhaustive optimization takes hours. 

We integrated all MOQO algorithms into the Postgres optimizer and present experimental results for TPC-H queries; we extended the Postgres cost model and optimize for up to nine conflicting objectives in our experiments. The proposed algorithms are based on a formal analysis of typical cost functions that occur in the context of MOQO. We identify properties that hold for a broad range of objectives and can be exploited for the design of future MOQO algorithms.

\end{abstract}

\input{Paper/sections/introSec.tex}

\section{Related Work}
\label{relatedSec}
\input{Paper/sections/related.tex}


\input{Paper/sections/modelSec.tex}

\input{Paper/sections/implementation.tex}

\input{Paper/sections/optimalEval.tex}

\input{Paper/sections/wMOQOsec}

\input{Paper/sections/bMOQOsec}

\input{Paper/sections/experimentalSec}

\section{Conclusion}
\label{conclusionSec}

\input{Paper/sections/conclusionSec.tex}


\balance

\begin{footnotesize}

\end{footnotesize}

\end{document}

%% file: Paper/sections/introSec.tex
\section{Introduction}
\label{introductionSec}

Minimizing execution time is the only objective in classical query optimization~\cite{Selinger1979}. Nowadays, there are however many scenarios in which additional objectives are of interest that should be considered during query optimization. This leads to the problem of multi-objective query optimization (MOQO) in which the goal is to find a query plan that realizes the best compromise between conflicting objectives. Consider the following example scenarios.

\begin{scenario}
A Cloud provider lets users submit SQL queries on data that resides in the Cloud. Queries are processed in the Cloud and users are billed according to the accumulated processing time over all nodes that participated in processing a certain query. The processing time of aggregation queries can be reduced by using sampling but this has a negative impact on result quality. From the perspective of the users, this leads to the three conflicting objectives of minimizing execution time, minimizing monetary costs, and minimizing the loss in result quality. Users specify preferences in their profiles by setting weights on different objectives, representing relative importance, and by optionally specifying constraints (e.g., an upper bound on execution time). Upon reception of a query, the Cloud provider needs to find a query plan that meets all constraints while minimizing the weighted sum over different cost metrics.
\end{scenario}

\begin{scenario}
A powerful server processes queries of multiple users concurrently. Minimizing the amount of system resources (such as buffer space, hard disk space, I/O bandwidth, and number of cores) that are dedicated for processing one specific query and minimizing that query's execution time are conflicting objectives (each specific system resource would correspond to an objective on its own). Upon reception of a query, the system must find a query plan that represents the best compromise between all conflicting objectives, considering weights and bounds defined by an administrator.
\end{scenario}

The main contribution in this paper are several MOQO algorithms that are generic enough to be applicable in a variety of scenarios (including the two scenarios outlined above) and are much more efficient than prior approaches while they formally guarantee to generate near-optimal query plans.

\subsection{State of the Art}

The goal of MOQO, according to our problem model, is to find query plans that minimize a weighted sum over different cost metrics while respecting all cost bounds. This means that multiple cost metrics are finally combined into a single metric (the weighted sum); it is still not possible to reduce MOQO to single-objective query optimization and use classic optimization algorithms such as the one by Selinger~\cite{Selinger1979}. Ganguly et al.\ have thoroughly justified why this is not possible~\cite{ganguly1992query}; we quickly outline the reasons in the following. Algorithms that prune plans based on a single cost metric must rely on the single-objective principle of optimality: replacing subplans (e.g., plans generating join operands) within a query plan by subplans that are better according to that cost metric cannot worsen the entire query plan according to that metric. This principle breaks when the cost metric of interest is a weighted sum over multiple metrics that are calculated according to diverse cost formulas. 

\begin{example}
Assume that each query plan is associated with a two-dimensional cost vector of the form $(t,e)$ where $t$ represents execution time in seconds and $e$ represents energy consumption in Joule. Assume one wants to minimize the weighted sum over time and energy with weight 1 for time and weight 2 for energy, i.e.\ the sum $t+2e$. Let $p$ be a plan that executes two subplans $p_1$ with cost vector $(7,1)$ and $p_2$ with cost vector $(6,2)$ in parallel. The cost vector of $p$ is $(7,3)$ since its execution time is the maximum over the execution times of its subplans ($7=\max(7,6)$) while its energy consumption is the sum of the energy consumptions of its subplans ($3=1+2$). Replacing $p_1$ within $p$ by another plan $p_1'$ with cost vector $(1,3)$ changes the cost vector of $p$ from $(7,3)$ to $(6,5)$. This means that the weighted cost of $p$ becomes worse (it increases from 13 to 16) even if the weighted cost of $p_1'$ (7) is better than the one of $p_1$ (9).
\end{example}

The example shows that the single-objective principle of optimality can break when optimizing a weighted sum of multiple cost metrics. Based on that insight, Ganguly et al.\ proposed a MOQO algorithm that uses a multi-objective version of the principle of optimality~\cite{ganguly1992query}. This algorithm guarantees to generate optimal query plans; it is however too computationally expensive for practical use as we will show in our experiments. The algorithm by Ganguly et al.\ is the only MOQO algorithm that we are aware of which is generic enough to handle all objectives that were mentioned in the example scenarios before. Most existing MOQO algorithms are specific to certain combinations of objectives where the single-objective principle of optimality holds~\cite{Agarwal2012, xu2012pet, Kambhampati02havasu:a}. 

\subsection{Contributions and Outline}

We summarize our contributions before we provide details:

\begin{itemize}
\item Our primary contribution are \textbf{two approximation schemes} for MOQO that scale to many objectives. They formally guarantee to return near-optimal query plans while speeding up optimization by several orders of magnitude in comparison with exact algorithms. 
\item We \textbf{formally analyze cost formulas} of many relevant objectives in query optimization and derive several common properties. We exploit these properties to design efficient approximation schemes and believe that our observations can serve as starting point for the design of future MOQO algorithms. 
\item We integrated the exact MOQO algorithm by Ganguly et al.~\cite{ganguly1992query} and our own MOQO approximation algorithms into the Postgres optimizer and \textbf{experimentally compare} their performance on TPC-H queries. 
\end{itemize}

Our approximation schemes formally guarantee to generate query plans whose cost is within a multiplicative factor $\alpha$ of the optimum in each objective. Parameter $\alpha$ can be tuned seamlessly to trade near-optimality guarantees for lower computational optimization cost. The near-optimality guarantees distinguish our approximation schemes from \textit{heuristics}, since heuristics can produce arbitrarily poor plans in the worst case. We show in our experimental evaluation that our approximation schemes reduce query optimization time from hours to seconds, comparing with an existing exact MOQO algorithm proposed by Ganguly et al.\ that is referred to as the \AG{} in the following. 

We discuss related work in Section~\ref{relatedSec} and introduce the formal model in Section~\ref{modelSec}. Our experimental evaluation is based on an extended version of Postgres that we describe in Section~\ref{implementationSec}. Note that our algorithms for MOQO are not specific to Postgres and can be used within any database system. We present the first experimental evaluation of the formerly proposed \AG{} in Section~\ref{optimalEvalSec}. Our experiments relate the poor scalability of the \AG{} to the high number of Pareto plans (i.e., plans representing an optimal tradeoff between different cost objectives) that it needs to generate. The representative-tradeoffs algorithm (\AWM{}), that we present in Section~\ref{wMOQOsec}, generates only one representative for multiple Pareto plans with similar cost tradeoffs and is therefore much more efficient than the \AG{}. We show that most common objectives in MOQO allow to construct near-optimal plans for joining a set of tables out of near-optimal plans for joining subsets. Due to that property, the \AWM{} formally guarantees to generate near-optimal query plans if user preferences are expressed by associating objectives with weights (representing relative importance). If users can specify cost bounds in addition to weights (representing for instance a monetary budget or a deadline), the \AWM{} cannot guarantee to generate near-optimal plans anymore and needs to be extended. We present the iterative-refinement algorithm (\ABM{}) in Section~\ref{bMOQOsec}. The \ABM{} uses the \AWM{} to generate a representative plan set in every iteration. The approximation precision is refined from one iteration to the next such that the representative plan set resembles more and more the Pareto plan set. The \ABM{} stops once it can guarantee that the generated plan set contains a near-optimal plan. A carefully selected precision refinement policy guarantees that the amount of redundant work (by repeatedly generating the same plans in different iterations) is negligible. We analyze time and space complexity of all presented algorithms and experimentally compare our two approximation schemes (the \AWM{} and the \ABM{}) against the \AG{} in Section~\ref{experimentalSec}.

%% file: Paper/sections/related.tex
Algorithms for \textbf{Single-Objective Query Optimization (SOQO)} are not applicable to MOQO or cannot offer any guarantees on result quality. Selinger et al.~\cite{Selinger1979} presented one of the first exact algorithms for SOQO which is based on dynamic programming. \textbf{Multi-Objective Query Optimization} is the focus of this paper. The algorithm by Ganguly et al.~\cite{ganguly1992query} is a generalization of the SOQO algorithm by Selinger et al. This algorithm is able to generate optimal query plans considering a multitude of objectives with diverse cost formulas. We describe it in more detail later, as we use it as baseline for our experiments. 

Algorithms for MOQO have not been experimentally evaluated for more than three objectives. They are usually tailored to very specific combinations of objectives. Neither the proposed algorithms nor the underlying algorithmic ideas can be used for many-objective QO with diverse cost formulas. Allowing only additive cost formulas (and user preference functions)~\cite{xu2012pet, Kambhampati02havasu:a} excludes for instance run time as objective in parallel execution scenarios (where time is calculated as maximum over parallel branches). The approach by Aggarwal et al.~\cite{Agarwal2012} is specific to the two objectives run time and confidence. Multiple objectives are only considered by selecting an optimal set of table samples prior to join ordering which does not generalize to different objectives. Optimizing different objectives separately misses optimal tradeoffs between conflicting objectives~\cite{abul2012alternative}. Separating join ordering and multi-objective optimization (e.g., by generating a time-optimal join tree first, and mapping join operators to sites considering multiple objectives later~\cite{garofalakis1996multi, Papadimitriou2001}) assumes that the same join tree is optimal for all objectives. This is only valid in special cases. Papadimitriou and Yannakakis~\cite{Papadimitriou2001} present multi-objective approximation algorithms for mapping operators to sites. Their algorithms do not optimize join order and the underlying approach does not generalize to more than one bounded objective. Algorithms for multi-objective optimization of data processing workflows~\cite{Simitsis2005, Simitsis2012, kllapi2011schedule} are not directly applicable to MOQO. Furthermore, the proposed approaches can be classified into heuristics that do not offer near-optimality guarantees~\cite{Simitsis2012, kllapi2011schedule}, and exact algorithms that do not scale~\cite{Simitsis2005}. 

\textbf{Parametric Query Optimization (PQO)} assumes that cost formulas depend on parameters with uncertain values. The goal is for instance to find robust plans~\cite{Babu2005, Babcock2005} or plans that optimize expected cost~\cite{Chu2002}. PQO and MOQO share certain problem properties while subtle differences prevent us from applying PQO algorithms to MOQO problems in general. Several approaches to PQO split for instance the PQO problem into several SOQO problems~\cite{Ganguly1998, Hulgeri2004, Bizarro2009} by fixing parameter values. This is not possible for MOQO since cost values, unlike parameter values, are only known once a query plan is complete and cannot be fixed in advance. Other PQO algorithms~\cite{Hulgeri2004} directly work with cost functions instead of scalar values during bottom-up plan construction. This assumes that all parameter values can be selected out of a connected interval which is typically not the case for cost objectives such as time or disc footprint. Our work connects to \textbf{Iterative Query Optimization} since we propose iterative algorithms. Kossmann and Stocker~\cite{Kossmann2000} propose several iterative algorithms that break the optimization of a large table set into multiple optimization runs for smaller table sets, thereby increasing efficiency. Their algorithm is only applicable to SOQO and does not offer formal guarantees on result quality. Work on \textbf{Skyline Queries}~\cite{kossmann2002shooting} and \textbf{Optimization Queries}~\cite{guha2003efficient} focuses on \textit{query processing} while we focus on \textit{query optimization}. Our work is situated in the broader area of \textbf{Approximation Algorithms}. We use generic techniques such as \textit{coarsening} that have been applied to other optimization problems~\cite{Erlebach2014, Marinescu2011}; the corresponding algorithms are however not applicable to query optimization and the specific coarsening methods differ. 




%% file: Paper/sections/modelSec.tex
\section{Formal Model}
\label{modelSec}

We represent \textbf{queries} as set of tables $Q$ that need to be joined. This model abstracts away details such as join predicates (that are however considered in the implementations of the presented algorithms). \textbf{Query plans} are characterized by the join order and the applied join and scan operators, chosen out of a set $\mathbb{J}$ of available operators. The two plans generating the inputs for the final join in a query plan $p$ are the \textbf{sub-plans} of $p$. The set $\mathbb{O}$ contains all \textbf{cost objectives} (e.g., $\mathbb{O}=\{\text{buffer space, execution time}\}$); we assume that a cost model is available for every objective that allows to estimate the cost of a plan. The function $\mathbf{c}(p)$ denotes the multi-dimensional cost of a plan $p$ (bold font distinguishes vectors from scalar values). Cost values are real-valued and non-negative. Let $o\in\mathbb{O}$ an objective, then $\mathbf{c}^o$ denotes the cost for $o$ within vector $\mathbf{c}$. Let $\mathbf{W}$ a vector of non-negative weights, then the function $C_{\mathbf{W}}(\mathbf{c})=\sum_{o\in\mathbb{O}}\mathbf{c}^o\mathbf{W}^o$ denotes the \textbf{weighted cost} of $\mathbf{c}$. Let $\mathbf{B}$ a vector of non-negative bounds (setting $\mathbf{B}^o=\infty$ means no bounds), then cost vector $\mathbf{c}$ \textbf{exceeds} the bounds if there is at least one objective $o$ with $\mathbf{c}^o>\mathbf{B}^o$. Vector $\mathbf{c}$ \textbf{respects} the bounds otherwise. The following two variants of the MOQO problem distinguish themselves by the expressiveness of the user preference model.


\begin{definition}
\textbf{Weighted MOQO Problem.}
A weighted MOQO problem instance is defined by a tuple $I=\langle Q,\mathbf{W}\rangle$ where $Q$ is a query and $\mathbf{W}$ a weight vector. A solution is a query plan for $Q$. An optimal plan minimizes the weighted cost $C_{\mathbf{W}}$ over all plans for $Q$. 
\end{definition}

\begin{definition}
\textbf{Bounded-Weighted MOQO Problem.}
A bounded-weighted MOQO problem instance is defined by a tuple $I=\langle Q,\mathbf{W},\mathbf{B}\rangle$ and extends the \PW{} problem by a bounds vector $\mathbf{B}$. Let $P$ the set of plans for $Q$ and $P_{\mathbf{B}}\subseteq P$ the set of plans that respect $\mathbf{B}$. If $P_{\mathbf{B}}$ is non-empty, an optimal plan minimizes $C_{\mathbf{W}}$ among the plans in $P_{\mathbf{B}}$. If $P_{\mathbf{B}}$ is empty, an optimal plan minimizes $C_{\mathbf{W}}$ among the plans in $P$. 
\end{definition}

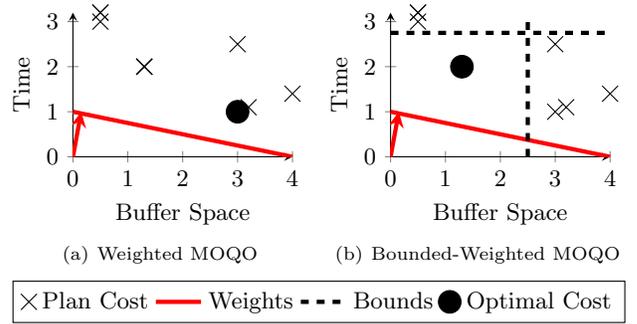
\begin{figure}
\subfigure[Weighted MOQO\label{WMOQOvarFig}]{\input{figures/diagrams/weightedMOQO1.tex}}
\subfigure[Bounded-Weighted MOQO\label{BMOQOvarFig}]{\input{figures/diagrams/boundedMOQO1.tex}}
\centering
\ref{variantsLegend}

\vspace{-3mm}

\caption{The two MOQO problem variants}
\label{variantsFig}
\end{figure}

Figure~\ref{WMOQOvarFig} illustrates \PW{}. It shows cost vectors of possible query plans (considering time and buffer space as objectives) and the user-specified weights (as vector from the origin). The line orthogonal to the weight vector represents cost vectors of equal weighted cost. The optimal plan is found by shifting this line to the top until it touches the first plan cost vector. Figure~\ref{BMOQOvarFig} illustrates \PB{}. Additional cost bounds are specified and a different plan is optimal since the formerly optimal plan exceeds the bounds. We will use the set of cost vectors depicted in Figure~\ref{variantsFig} as \textbf{running example} throughout the paper. The relative cost function $\rho$ measures the cost of a plan relative to an optimal plan. 

\begin{definition}
\textbf{Relative Cost.}
\label{relCostDef}
The relative cost function $\rho_I$ of a \PW{} instance $I=\langle Q,\mathbf{W}\rangle$ judges a query plan $p$ by comparing its weighted cost to the one of an optimal plan $p^*$: $\rho_I(p)=C_{\mathbf{W}}(\mathbf{c}(p))/C_{\mathbf{W}}(\mathbf{c}(p^*))$. The relative cost function of a \PB{} instance $I=\langle Q,\mathbf{W},\mathbf{B}\rangle$ is defined in the same way if no plan exists that respects $\mathbf{B}$. Otherwise, set $\rho_I(p)=\infty$ for any plan $p$ that does not respect $\mathbf{B}$ and $\rho_I(p)=C_{\mathbf{W}}(\mathbf{c}(p))/C_{\mathbf{W}}(\mathbf{c}(p^*))$ if $p$ respects $\mathbf{B}$. 
\end{definition}

Let $\alpha\geq1$, then an $\alpha$-\textbf{approximate solution} to a \PW{} or \PB{} instance $I$ is a plan $p$ whose relative cost is bounded by $\alpha$: $\rho_I(p)\leq\alpha$. The following classification of MOQO algorithms is based on the formal near-optimality guarantees that they offer. 

\begin{definition}
\textbf{MOQO Approximation Scheme.}
\label{approximationAlgDef}
An approximation scheme for MOQO is tuned via a user-specified precision parameter $\alpha_U$ and guarantees to generate an $\alpha_U$-approximate solution for any MOQO problem instance. 
\end{definition}

\begin{definition}
\textbf{Exact MOQO Algorithm.}
\label{optimalAlgDef}
An exact algorithm for MOQO guarantees to generate a $1$-approximate (hence optimal) solution for any MOQO problem instance. 
\end{definition}

\begin{figure}
\centering
\input{figures/diagrams/pareto1.tex}

\vspace{-5mm}

\caption{Pareto frontier and dominated area}
\label{paretoFig}
\end{figure}
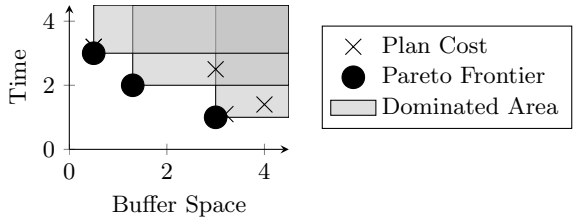

The following definitions express relationships between cost vectors. A vector $\mathbf{c}_1$ \textbf{dominates} vector $\mathbf{c}_2$, denoted by $\mathbf{c}_1\preceq\mathbf{c}_2$, if $\mathbf{c}_1$ has lower or equivalent cost than $\mathbf{c}_2$ in every objective. Vector $\mathbf{c}_1$ \textbf{strictly dominates} $\mathbf{c}_2$, denoted by $\mathbf{c}_1\prec\mathbf{c}_2$, if $\mathbf{c}_1\preceq\mathbf{c}_2$ and the vectors are not equivalent ($\mathbf{c}_1\neq\mathbf{c}_2$). Vector $\mathbf{c}_1$ \textbf{approximately dominates} $\mathbf{c}_2$ with precision $\alpha$, denoted by $\mathbf{c}_1\preceq_{\alpha}\mathbf{c}_2$, if the cost of $\mathbf{c}_1$ is higher at most by factor $\alpha$ in every objective, i.e. $\forall o:\mathbf{c}_1^o\leq\mathbf{c}_2^o\cdot\alpha$. A plan $p$ and its cost vector are \textbf{Pareto-optimal} for query $Q$ (short: \textbf{Pareto plan} and \textbf{Pareto vector}) if no alternative plan for $Q$ strictly dominates $p$. A \textbf{Pareto set} for $Q$ contains at least one cost-equivalent plan for each Pareto plan. The \textbf{Pareto frontier} is the set of all Pareto vectors. 
Figure~\ref{paretoFig} shows the Pareto frontier of the running example and the area that each Pareto vector dominates. An $\alpha$\textbf{-approximate Pareto set} for $Q$ contains for every Pareto plan $p^*$ a plan $p$ such that $\mathbf{c}(p)\preceq_{\alpha}\mathbf{c}(p^*)$. An $\alpha$\textbf{-approximate Pareto frontier} contains the cost vectors of all plans in an $\alpha$-approximate Pareto set. 
During complexity analysis, $j=|\mathbb{J}|$ denotes the number of operators, $l=|\mathbb{O}|$ the number of objectives, $n=|Q|$ the number of tables to join, and $m$ the maximal cardinality over all base tables in the database. Users formulate queries and have direct influence on table cardinalities. Therefore, $n$ and $m$ (and also $j$) are treated as variables during asymptotic analysis. Introducing new objectives (that cannot be derived from existing ones) requires changes to the code base and detailed experimental analysis to provide realistic cost formulas. This is typically not done by users, therefore $l$ is treated as a constant (it is common to treat the number of objectives as a constant when analyzing multi-objective approximation schemes~\cite{Papadimitriou2001, Erlebach2014}).

%% file: figures/diagrams/weightedMOQO1.tex
\pgfplotstableread{figures/diagrams/points.txt}\points
\pgfplotstableread{figures/diagrams/pareto.txt}\pareto
\pgfplotstableread{figures/diagrams/fine.txt}\fine
\pgfplotstableread{figures/diagrams/coarse.txt}\coarse

\begin{tikzpicture}

\begin{axis}[
xlabel=\reMetricOne{},
ylabel=\reMetricTwo{},
ylabel near ticks,
xlabel near ticks,
axis x line=bottom,
axis y line=left,
xmin=0,
ymin=0,
mark=x,
]

\addplot[only marks,mark=x] table[x={x},y={y}]{\points};
\addplot[only marks,forget plot] table[x={x},y={y}]{\coarse};
\addplot[only marks,forget plot] table[x={x},y={y}]{\fine};
\addplot[only marks,forget plot] table[x={x},y={y}]{\pareto};

\draw[red,-stealth,ultra thick]
(axis cs:0,0)
--
++
(axis direction cs:0.15,1);

\addplot[no marks,draw=red,ultra thick] coordinates {
(0,1)
(4,0)
};

\addplot[mark=*,fill=black,only marks] coordinates {
(3,1)
};

\end{axis}

\end{tikzpicture}

%% file: figures/diagrams/boundedMOQO1.tex
\pgfplotstableread{figures/diagrams/points.txt}\points
\pgfplotstableread{figures/diagrams/pareto.txt}\pareto
\pgfplotstableread{figures/diagrams/fine.txt}\fine
\pgfplotstableread{figures/diagrams/coarse.txt}\coarse

\begin{tikzpicture}

\begin{axis}[
xlabel=\reMetricOne{},
ylabel=\reMetricTwo{},
ylabel near ticks,
xlabel near ticks,
axis x line=bottom,
axis y line=left,
xmin=0,
ymin=0,
mark=x,
legend columns=4,
legend style={
at={(1,-0.25)},
anchor=north east},
legend to name=variantsLegend,
legend cell align=left]

\addplot[only marks,mark=x] table[x={x},y={y}]{\points};
\addplot[only marks,forget plot] table[x={x},y={y}]{\coarse};
\addplot[only marks,forget plot] table[x={x},y={y}]{\fine};
\addplot[only marks,forget plot] table[x={x},y={y}]{\pareto};

\draw[red,-stealth,ultra thick]
(axis cs:0,0)
--
++
(axis direction cs:0.15,1);

\addplot[no marks,draw=red,ultra thick] coordinates {
(0,1)
(4,0)
};

\addplot[ultra thick,dashed,no marks] coordinates {
(2.5,0)
(2.5,3)
};
\addplot[ultra thick,dashed,no marks,forget plot] coordinates {
(0,2.75)
(4,2.75)
};

\addplot[mark=*,fill=black,only marks] coordinates {
(1.3,2)
};

\addlegendentry[align=left]{\planCost{}}
\addlegendentry[align=left]{\weights{}}
\addlegendentry[align=left]{\bounds{}}
\addlegendentry[align=left]{\optimalCost{}}

\end{axis}

\end{tikzpicture}

%% file: figures/diagrams/pareto1.tex
\pgfplotstableread{figures/diagrams/points.txt}\points
\pgfplotstableread{figures/diagrams/pareto.txt}\pareto
\pgfplotstableread{figures/diagrams/fine.txt}\fine
\pgfplotstableread{figures/diagrams/coarse.txt}\coarse

\begin{tikzpicture}

\begin{axis}[
xlabel=\reMetricOne{},
ylabel=\reMetricTwo{},
ylabel near ticks,
xlabel near ticks,
axis x line=bottom,
axis y line=left,
xmin=0,
ymin=0,
mark=x,
legend columns=1,
legend style={
at={(1.15,0.5)},
anchor=west},
legend cell align=left]

\addplot[only marks,mark=x] table[x={x},y={y}]{\points};
\addplot[only marks,forget plot] table[x={x},y={y}]{\coarse};
\addplot[only marks,forget plot] table[x={x},y={y}]{\fine};
\addplot[only marks,mark=*,fill=black] table[x={x},y={y}]{\pareto};

\foreach \i in {2,1,0} {

\pgfplotstablegetelem{\i}{x}\of{\pareto}
\let\xCoord=\pgfplotsretval

\pgfplotstablegetelem{\i}{y}\of{\pareto}
\let\yCoord=\pgfplotsretval

\addplot[const plot mark right,fill=lightgray,no markers, fill opacity=0.5, area legend] coordinates {
(\xCoord,\yCoord)
(4.5,\yCoord)
(4.5,4.5)
(\xCoord,4.5)
(\xCoord,\yCoord)
};

} 

\addlegendentry[align=left]{\planCost{}}
\addlegendentry[align=left]{\paretoFrontier{}}
\addlegendentry[align=left]{\dominated{}}

\end{axis}

\end{tikzpicture}

%% file: Paper/sections/implementation.tex
\section{Prototypical Implementation}
\label{implementationSec}

\newcommand{\LineitemLabel}{L}
\newcommand{\OrdersLabel}{O}
\newcommand{\CustomerLabel}{C}

\newcommand{\HJLabel}{HashJ}
\newcommand{\SMJLabel}{SMJ}
\newcommand{\INLLabel}{IdxNL}


\label{costSub}

We extended the Postgres system (version~9.2.4) to obtain an experimental platform for comparing MOQO algorithms. We extended the cost model, the query optimizer, and the user interface.
The extended cost model supports nine objectives. The cost formulas used in the cost model are taken from prior work and are not part of our contribution. Evaluating their accuracy is beyond the scope of this paper. We quickly describe the nine implemented cost objectives. \textbf{Total execution time} (i.e., time until all result tuples have been produced) and \textbf{startup time} (i.e., time until first result tuple is produced) are estimated according to the cost formulas already included in Postgres. Minimizing \textbf{IO load}, \textbf{CPU load}, \textbf{number of used cores}, \textbf{hard disc footprint}, and \textbf{buffer footprint} is important since it allows to increase the number of concurrent users. The five aforementioned objectives often conflict with run time since using more system resources can often speed up query processing. \textbf{Energy consumption} is not always correlated with time~\cite{xu2012pet, Flach2010}. Dedicating more cores to a query plan can for instance decrease execution time by parallelization while it introduces coordination overhead that results in higher total energy consumption. Energy consumption is calculated according to the cost formulas by Flach~\cite{Flach2010}. Sampling allows to trade result completeness for efficiency~\cite{Haas2003}. The \textbf{tuple loss} ratio is the expected fraction of lost result tuples due to sampling and serves as ninth objective. Joining two operands with tuple loss $a,b\in[0,1]$, the tuple loss of the result is estimated by the formula $1-(1-a)(1-b)$. 

\label{planSpaceSub}

We extended the plan space of the Postgres optimizer by introducing new operators and parameterizing existing ones (we did not implement those operators in the execution engine). The extended plan space includes a parameterized sampling operator that scans between 1\% and 5\% of a base table. Join and sort operators are parameterized by the degree of parallelism (DOP). The DOP represents the number of cores that process the corresponding operation (up to 4 cores can be used per operation). 
The Postgres optimizer uses several heuristics to restrict the search space: in particular, \textit{i)}~it considers Cartesian products only in situations in which no other join is applicable, and \textit{ii)}~it optimizes different subqueries of the same query separately. We left both heuristics in place since removing them might  have significant impact on performance. Not using those heuristics would make it difficult to decide whether high computational costs observed during MOQO are due to the use of multiple objectives or to the removal of the heuristics. 

\begin{figure}[t]
\centering
L=Lineitem; O=Orders; C=Customers; \HJLabel=Hash Join; \SMJLabel=Sort-Merge Join; \INLLabel=Index-Nested-Loop Join

\subfigure[Time-Optimal Plan for Bounded Tuple Loss ($=0$)\label{planASub}]{\makebox[0.3\columnwidth][c]{\input{figures/queryPlans/planA1.tex}}}
\hfill
\subfigure[Additional Weight on Buffer Space Leads to Plan Without Hash Joins\label{planBSub}]{\makebox[0.3\columnwidth][c]{\input{figures/queryPlans/planB1.tex}}}
\hfill
\subfigure[Additional Bound on Startup Time Requires Using Nested-Loop Joins\label{planCSub}]{\makebox[0.3\columnwidth][c]{\input{figures/queryPlans/planC1.tex}}}

\vspace{-3mm}

\caption{Evolution of optimal plan for TPC-H Query 3 when changing user preferences\label{queryPlansFig}}
\end{figure}
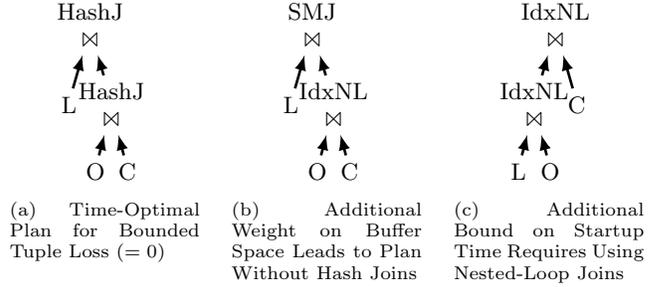

The original Postgres optimizer is single-objective and optimizes total execution time. We implemented all three MOQO algorithms that are discussed in this paper: the \AG{}, the \AWM{}, and the \ABM{}. The implementation uses the original Postgres data structures and routines wherever possible. Users can switch between the optimization algorithms and can choose the approximation precision $\alpha$ for the two approximation schemes. 
Users can specify weights and bounds on the different objectives. The higher the weight on some objective, the higher its relative importance. Bounds allow to specify cost limits for specific objectives (e.g., time limits or energy budgets). When optimizing a query, the optimizer tries to find a plan that minimizes the weighted cost among all plans that respect the bounds. Figure~\ref{queryPlansFig} shows how the optimal query plan for TPC-H query 3 changes when user preferences vary. Initially, the tuple loss is upper-bounded by zero (i.e., all result tuples must be retrieved) and all weights except the one for total execution time are set to zero. So the optimizer searches for the plan with minimal execution time among all plans that do not use sampling. Figure~\ref{planASub} shows the resulting plan. Increasing the weight on buffer footprint leads to a plan that replaces the memory-intensive Hash joins by Sort-Merge and Index-Nested-Loop (IdxNL) joins (see Figure~\ref{planBSub}). Setting an additional upper bound on startup time leads to a plan that only uses IdxNL joins (see Figure~\ref{planCSub}). 

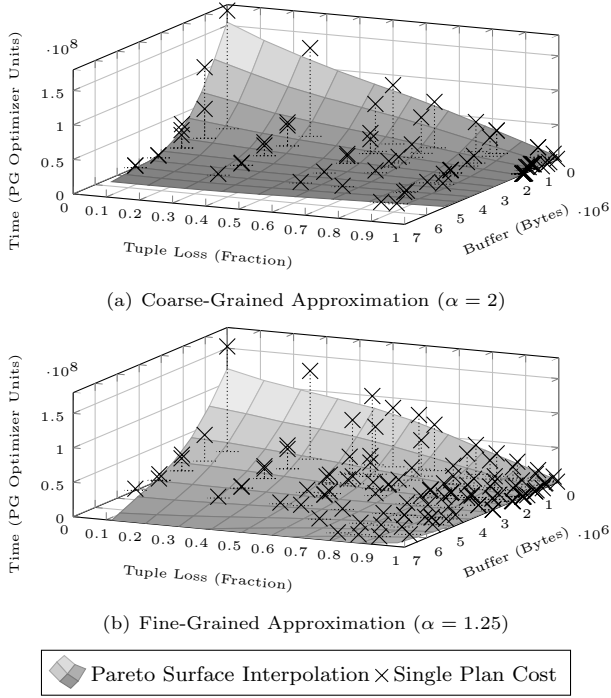
\begin{figure}[t]
\centering
\subfigure[Coarse-Grained Approximation ($\alpha=2$)\label{tpcParetoCoarse}]{%
\begin{tikzpicture}
    \begin{axis}[grid=major,y dir=reverse, xlabel=Tuple Loss (Fraction),ylabel=Buffer (Bytes),zlabel=Time (PG Optimizer Units),tiny,height=4.5cm,width=8cm,
xlabel style={sloped like x axis},ylabel style={sloped like y axis},zlabel style={sloped like z axis},
colormap={bw}{gray(0cm)=(0.5); gray(1cm)=(1)},
ymin=0,ymax=7e6,zmin=0,zmax=18e7]
        \addplot3[surf,shader=faceted,mesh/ordering=y varies,mesh/rows=10,mesh/cols=10] file {figures/diagrams/tpch5_c2_LBT_surface.dat};
	\addplot3[only marks,mark=x,color=black, 
	error bars/.cd,z dir=minus,z fixed relative=1, error bar style={densely dotted}] file {figures/diagrams/tpch5_c2_LBT_points.dat};
    \end{axis}
\end{tikzpicture}}
\subfigure[Fine-Grained Approximation ($\alpha=1.25$)\label{}\label{tpcParetoFine}]{%
\begin{tikzpicture}
    \begin{axis}[grid=major,y dir=reverse, xlabel=Tuple Loss (Fraction),ylabel=Buffer (Bytes),zlabel=Time (PG Optimizer Units),tiny,height=4.5cm,width=8cm,
xlabel style={sloped like x axis},ylabel style={sloped like y axis},zlabel style={sloped like z axis},
colormap={bw}{gray(0cm)=(0.5); gray(1cm)=(1)},
ymin=0,ymax=7e6,zmin=0,zmax=18e7,
legend to name=TPCH5_Pareto_Legend,legend style={legend columns=2,font=\small}]
        \addplot3[colormap={bw}{gray(0cm)=(0.5); gray(1cm)=(1)},surf,shader=faceted,mesh/ordering=y varies,mesh/rows=10,mesh/cols=10] file {figures/diagrams/tpch5_c125_LBT_surface.dat};
	\addplot3[only marks,mark=x,color=black, 
	error bars/.cd,z dir=minus,z fixed relative=1, error bar style={densely dotted}] file {figures/diagrams/tpch5_c125_LBT_points.dat};
	
	\legend{Pareto Surface Interpolation, Single Plan Cost}
	
    \end{axis}
\end{tikzpicture}}
\ref{TPCH5_Pareto_Legend}

\vspace{-3mm}

\caption{Three-dimensional Pareto frontier approximations for \mbox{TPC-H} Query 5\label{tpcParetoFigure}}
\end{figure}

Users cannot make optimal choices for bounds and weights if they are not aware of the possible tradeoffs between different objectives. A user might for instance want to relax the bound on one objective, knowing that this allows significant savings in another objective. All implemented MOQO algorithms produce an (approximate) Pareto frontier as byproduct of optimization. Our prototype allows to visualize two and three dimensional projections of the Pareto frontier. Figure~\ref{tpcParetoFigure} shows the cost vectors of the approximate Pareto frontier for TPC-H query~5 (and an interpolation of the surface defined by those vectors), considering objectives tuple loss, buffer footprint, and total execution time. Figure~\ref{tpcParetoCoarse} shows a coarse-grained approximation of the real Pareto frontier (with $\alpha=2$) and Figure~\ref{tpcParetoFine} a more fine-grained approximation for the same query ($\alpha=1.25$). 

%% file: figures/queryPlans/planA1.tex
\begin{tikzpicture}[>=latex,line join=bevel,]
  \pgfsetlinewidth{1bp}
\pgfsetcolor{black}
  \draw [<-] (15.888bp,49.891bp) .. controls (17.999bp,42.502bp) and (18.044bp,42.345bp)  .. (18.089bp,42.188bp);
  \draw [<-] (18.351bp,19.845bp) .. controls (16.396bp,12.35bp) and (16.358bp,12.205bp)  .. (16.32bp,12.062bp);
  \draw [<-] (10.112bp,49.891bp) .. controls (7.0431bp,39.151bp) and (6.7676bp,38.186bp)  .. (6.5141bp,37.3bp);
  \draw [<-] (23.649bp,19.845bp) .. controls (25.604bp,12.35bp) and (25.642bp,12.205bp)  .. (25.68bp,12.062bp);
\begin{scope}
  \definecolor{strokecol}{rgb}{0.0,0.0,0.0};
  \pgfsetstrokecolor{strokecol}
  \draw (15bp,6bp) node {\OrdersLabel};
\end{scope}
\begin{scope}
  \definecolor{strokecol}{rgb}{0.0,0.0,0.0};
  \pgfsetstrokecolor{strokecol}
  \draw (27bp,6bp) node {\CustomerLabel};
\end{scope}
\begin{scope}
  \definecolor{strokecol}{rgb}{0.0,0.0,0.0};
  \pgfsetstrokecolor{strokecol}
  \draw (5bp,31bp) node {\LineitemLabel};
\end{scope}
\begin{scope}
  \definecolor{strokecol}{rgb}{0.0,0.0,0.0};
  \pgfsetstrokecolor{strokecol}
  \draw (13bp,61bp) node {\parbox{0.5in}{\centering \HJLabel\\$\Join$}};
\end{scope}
\begin{scope}
  \definecolor{strokecol}{rgb}{0.0,0.0,0.0};
  \pgfsetstrokecolor{strokecol}
  \draw (21bp,31bp) node {\parbox{0.5in}{\centering \HJLabel\\$\Join$}};
\end{scope}
\end{tikzpicture}

%% file: figures/queryPlans/planB1.tex
\begin{tikzpicture}[>=latex,line join=bevel,]
  \pgfsetlinewidth{1bp}
\pgfsetcolor{black}
  \draw [<-] (15.888bp,49.891bp) .. controls (17.999bp,42.502bp) and (18.044bp,42.345bp)  .. (18.089bp,42.188bp);
  \draw [<-] (18.351bp,19.845bp) .. controls (16.396bp,12.35bp) and (16.358bp,12.205bp)  .. (16.32bp,12.062bp);
  \draw [<-] (10.112bp,49.891bp) .. controls (7.0431bp,39.151bp) and (6.7676bp,38.186bp)  .. (6.5141bp,37.3bp);
  \draw [<-] (23.649bp,19.845bp) .. controls (25.604bp,12.35bp) and (25.642bp,12.205bp)  .. (25.68bp,12.062bp);
\begin{scope}
  \definecolor{strokecol}{rgb}{0.0,0.0,0.0};
  \pgfsetstrokecolor{strokecol}
  \draw (15bp,6bp) node {\OrdersLabel};
\end{scope}
\begin{scope}
  \definecolor{strokecol}{rgb}{0.0,0.0,0.0};
  \pgfsetstrokecolor{strokecol}
  \draw (27bp,6bp) node {\CustomerLabel};
\end{scope}
\begin{scope}
  \definecolor{strokecol}{rgb}{0.0,0.0,0.0};
  \pgfsetstrokecolor{strokecol}
  \draw (5bp,31bp) node {\LineitemLabel};
\end{scope}
\begin{scope}
  \definecolor{strokecol}{rgb}{0.0,0.0,0.0};
  \pgfsetstrokecolor{strokecol}
  \draw (13bp,61bp) node {\parbox{0.5in}{\centering \SMJLabel\\ $\Join$}};
\end{scope}
\begin{scope}
  \definecolor{strokecol}{rgb}{0.0,0.0,0.0};
  \pgfsetstrokecolor{strokecol}
  \draw (21bp,31bp) node {\parbox{0.5in}{\centering \INLLabel\\$\Join$}};
\end{scope}
\end{tikzpicture}

%% file: figures/queryPlans/planC1.tex
\begin{tikzpicture}[>=latex,line join=bevel,]
  \pgfsetlinewidth{1bp}
\pgfsetcolor{black}
  \draw [<-] (16.112bp,49.891bp) .. controls (14.001bp,42.502bp) and (13.956bp,42.345bp)  .. (13.911bp,42.188bp);
  \draw [<-] (8.3508bp,19.845bp) .. controls (6.3957bp,12.35bp) and (6.3578bp,12.205bp)  .. (6.3204bp,12.062bp);
  \draw [<-] (21.888bp,49.891bp) .. controls (24.957bp,39.151bp) and (25.232bp,38.186bp)  .. (25.486bp,37.3bp);
  \draw [<-] (13.649bp,19.845bp) .. controls (15.604bp,12.35bp) and (15.642bp,12.205bp)  .. (15.68bp,12.062bp);
\begin{scope}
  \definecolor{strokecol}{rgb}{0.0,0.0,0.0};
  \pgfsetstrokecolor{strokecol}
  \draw (5bp,6bp) node {\LineitemLabel};
\end{scope}
\begin{scope}
  \definecolor{strokecol}{rgb}{0.0,0.0,0.0};
  \pgfsetstrokecolor{strokecol}
  \draw (17bp,6bp) node {\OrdersLabel};
\end{scope}
\begin{scope}
  \definecolor{strokecol}{rgb}{0.0,0.0,0.0};
  \pgfsetstrokecolor{strokecol}
  \draw (27bp,31bp) node {\CustomerLabel};
\end{scope}
\begin{scope}
  \definecolor{strokecol}{rgb}{0.0,0.0,0.0};
  \pgfsetstrokecolor{strokecol}
  \draw (19bp,61bp) node {\parbox{0.5in}{\centering \INLLabel\\ $\Join$}};
\end{scope}
\begin{scope}
  \definecolor{strokecol}{rgb}{0.0,0.0,0.0};
  \pgfsetstrokecolor{strokecol}
  \draw (11bp,31bp) node {\parbox{0.5in}{\centering \INLLabel\\ $\Join$}};
\end{scope}
\end{tikzpicture}

%% file: Paper/sections/optimalEval.tex
\section{Analysis of Exact Algorithm}
\label{optimalEvalSec}

\begin{algorithm}[t!]
\renewcommand{\algorithmiccomment}[1]{// #1}
\begin{small}
\begin{algorithmic}[1]
\State \Comment{Find best plan for query $Q$, weights $\mathbf{W}$, bounds $\mathbf{B}$}
\Function{ExactMOQO}{$Q,\mathbf{W},\mathbf{B}$}
\State \Comment{Find Pareto plan set for $Q$}
\State $\mathcal{P}\leftarrow\mathrm{FindParetoPlans}(Q)$
\State \Comment{Return best plan out of Pareto plans}
\State \Return SelectBest$(\mathcal{P},\mathbf{W},\mathbf{B})$
\EndFunction
\vspace{0.25cm}
\State \Comment{Find Pareto plan set for query $Q$}
\Function{FindParetoPlans}{$Q$}
\State \Comment{Calculate plans for singleton sets}
\ForAll{$q\in Q$} 
\State $\mathcal{P}^q\leftarrow\emptyset$
\ForAll{$p_N$ access path for $q$}
\State $\mathrm{Prune}(\mathcal{P}^q,p_N)$
\EndFor
\EndFor
\State \Comment{Consider table sets of increasing cardinality}
\ForAll{$k\in2..|Q|$}
\ForAll{$q\subseteq Q:|q|=k$}
\State $\mathcal{P}^{q}\leftarrow\emptyset$
\State \Comment{For all possible splits of set $q$}
\ForAll{$q_{1},q_{2}\subset q:q_{1}\dot{\cup}q_{2}=q$}
\State \Comment{For all sub-plans and operators}
\ForAll{$p_{1}\in\mathcal{P}^{q_1},p_{2}\in\mathcal{P}^{q_2},j\in\mathbb{J}$}
\State \Comment{Construct new plan out of sub-plans}
\State $p_N\leftarrow\mathrm{Combine}(j,p_{1},p_{2})$
\State \Comment{Prune with new plan}
\State $\mathrm{Prune}(\mathcal{P}^q,p_N)$
\EndFor
\EndFor
\EndFor
\EndFor
\State \Return{$\mathcal{P}^Q$}
\EndFunction
\vspace{0.25cm}
\State \Comment{Prune plan set $\mathcal{P}$ with new plan $p_N$}
\Procedure{Prune}{$\mathcal{P},p_N$}
\State \Comment{Check whether new plan useful}
\If{$\neg\exists p\in\mathcal{P}:\mathbf{c}(p)\preceq\mathbf{c}(p_N)$}
\State \Comment{Delete dominated plans}
\State $\mathcal{P}\leftarrow\{p\in\mathcal{P}\mid\neg(\mathbf{c}(p_N)\preceq\mathbf{c}(p))\}$
\State \Comment{Insert new plan}
\State $\mathcal{P}\leftarrow\mathcal{P}\cup\{p_N\}$
\EndIf
\EndProcedure
\vspace{0.25cm}
\State \Comment{Select best plan in $\mathcal{P}$ for weights $\mathbf{W}$ and bounds $\mathbf{B}$}
\Function{SelectBest}{$\mathcal{P},\mathbf{W},\mathbf{B}$}
\State $P_{\mathbf{B}}\leftarrow\{p\in P\mid\mathbf{c}(p)\preceq\mathbf{B}\}$
\If{$P_{\mathbf{B}}\neq\emptyset$}
\State \Return $\arg\min[p\in P_{\mathbf{B}}]C_{\mathbf{W}}(\mathbf{c}(p))$
\Else
\State \Return $\arg\min[p\in P]C_{\mathbf{W}}(\mathbf{c}(p))$
\EndIf
\EndFunction
\end{algorithmic}
\end{small}
\caption{Exact algorithm for MOQO\label{exactAlg}}
\end{algorithm}

Ganguly et al.~\cite{ganguly1992query} proposed an exact algorithm (\AG{}) for MOQO. This algorithm is not part of our contribution but we provide a first experimental evaluation in a many-objective scenario and a formal analysis under less optimistic assumptions than in the original publication. Algorithm~\ref{exactAlg} shows the pseudo-code of the \AG{} (compared with the original publication, the code was slightly extended to generate bushy plans in addition to left-deep plans). The \AG{} first calculates a Pareto plan set for query $Q$ and finally selects the optimal plan out of that set (considering weights and bounds). The \AG{} uses dynamic programming and constructs Pareto plans for a table set out of the Pareto plans of its subsets. It is a generalization of the seminal algorithm by Selinger et al.~\cite{Selinger1979}, generalizing the pruning metric from one to multiple cost objectives. The \AG{} starts by calculating Pareto plans for single tables. Plans generating the same result are compared and \textit{pruned}, meaning that dominated plans are discarded. The \AG{} constructs Pareto plans for table sets of increasing cardinality. To generate plans for a specific table set, the \AG{} considers \textit{i)}~all possible splits of that set into two non-empty subsets (every split corresponds to one choice of operands for the last join), \textit{ii)}~all available join operators, and \textit{iii)}~all combinations of Pareto plans for generating the two inputs to the last join. 

\subsection{Experimental Analysis}
\label{sub:optimalExperiments}

We implemented the \AG{} within the system described in Section~\ref{implementationSec}. The implementation allows to specify timeouts (the corresponding code is not shown in Algorithm~\ref{exactAlg}). If the optimization time exceeds two hours, the modified \AG{} finishes quickly by only generating one plan for all table sets that have not been treated so far. We experimentally evaluated the \AG{} using the \mbox{TPC-H}~\cite{TPC2013} benchmark. We generated several test cases for each TPC-H query by randomly selecting subsets of objectives with a fixed cardinality out of the total set of nine objectives. All experiments were executed on a server equipped with two six core Intel Xeon processors with 2~GhZ and 128~GB of DDR3 RAM running Linux 2.6 (64 bit version). Five optimizer threads ran in parallel during the experiments. 

The goal of the evaluation was to answer three questions: \textit{i)}~Is the performance of the \AG{} good enough for use in practice? \textit{ii)}~If not, how can the performance be improved? \textit{iii)}~What assumptions are realistic for the formal complexity analysis of MOQO algorithms? Figure~\ref{optimalEvalFIg} shows experimental results for the three metrics optimization time, allocated memory during optimization, and number of Pareto plans for the last table set that was treated completely (before a timeout occurred or before the optimization was completed). Every marker represents the arithmetic average value over 20 test cases for one specific \mbox{TPC-H} query and a specific number of objectives. The \mbox{TPC-H} queries are ordered according to the maximal number of tables that appears in any of their from-clauses. This number correlates (with several caveats\footnote{The Postgres optimizer may for instance convert EXISTS predicates into joins which leads to many alternative plans even for queries with only one table in the from-clause. }) with the search space size. Gray markers indicate that some test cases incurred a timeout. If a timeout occurred, then the reported values are lower bounds on the values of a completed computation. 

\begin{figure}
\begin{tikzpicture}

\begin{groupplot}[tiny,
group style={ %
  group size=2 by 2, vertical sep=0cm},
  scatter/classes={1={mark=x,black,mark size=2},2={mark=x,lightgray,mark size=2},
  31={mark=*,black,mark size=2},32={mark=*,lightgray,draw=black,mark size=2},
  61={mark=triangle*,black,mark size=2},62={mark=triangle*,lightgray,draw=black,mark size=2},
  91={mark=square*,black,mark size=2},92={mark=square*,lightgray,draw=black,mark size=2}},
ymajorgrids,xmajorgrids]

\nextgroupplot[ylabel=Time (ms),ymode=log,only marks,
width=\columnwidth*0.54, extra y ticks={7200000}, extra y tick label={Out}, ytick={10, 100, 1000, 10000, 100000, 1000000}, ymax=20000000,xticklabels={}]
\addplot[scatter,scatter src=explicit symbolic] table[header=false, x index=0, y index=1, meta index=2] {benchmark/results/\oneObjWeighted/timeOpt.txt};
\addplot[scatter,scatter src=explicit symbolic] table[header=false, x index=0, y index=1, meta index=2] {benchmark/results/\threeObjWeighted/timeOpt.txt};
\addplot[scatter,scatter src=explicit symbolic] table[header=false, x index=0, y index=1, meta index=2] {benchmark/results/\sixObjWeighted/timeOpt.txt};
\addplot[scatter,scatter src=explicit symbolic] table[header=false, x index=0, y index=1, meta index=2] {benchmark/results/\nineObjWeighted/timeOpt.txt};

\nextgroupplot[xlabel=\# Join Tables,ylabel=Memory (KB),ymode=log,only marks,
width=\columnwidth*0.54, ytick={1000, 10000, 100000, 1000000, 10000000},ylabel shift=-0.15cm]
\addplot[scatter,scatter src=explicit symbolic] table[header=false, x index=0, y index=1, meta index=2] {benchmark/results/\oneObjWeighted/\memoryOpt.txt};
\addplot[scatter,scatter src=explicit symbolic] table[header=false, x index=0, y index=1, meta index=2] {benchmark/results/\threeObjWeighted/\memoryOpt.txt};
\addplot[scatter,scatter src=explicit symbolic] table[header=false, x index=0, y index=1, meta index=2] {benchmark/results/\sixObjWeighted/\memoryOpt.txt};
\addplot[scatter,scatter src=explicit symbolic] table[header=false, x index=0, y index=1, meta index=2] {benchmark/results/\nineObjWeighted/\memoryOpt.txt};

\nextgroupplot[xlabel=\# Join Tables,ylabel=\# Pareto Plans,ymode=log,only marks,
width=\columnwidth*0.54,ytick={1, 8, 64, 512, 5000},yticklabels={1,8,64,512, 5000, $10^4$},
legend to name=optimalEvalLegend]
\addplot[scatter,scatter src=explicit symbolic] table[header=false, x index=0, y index=1, meta index=2] {benchmark/results/\oneObjWeighted/pathsOpt.txt};
\addplot[scatter,scatter src=explicit symbolic] table[header=false, x index=0, y index=1, meta index=2] {benchmark/results/\threeObjWeighted/pathsOpt.txt};
\addplot[scatter,scatter src=explicit symbolic] table[header=false, x index=0, y index=1, meta index=2] {benchmark/results/\sixObjWeighted/pathsOpt.txt};
\addplot[scatter,scatter src=explicit symbolic] table[header=false, x index=0, y index=1, meta index=2] {benchmark/results/\nineObjWeighted/pathsOpt.txt};


\legend{{1 Objective, No Timeouts}, {1 Objective, Timeouts}, {3 Objectives, No Timeouts}, {3 Objectives, Timeouts}, {6 Objectives, No Timeouts}, {6 Objectives, Timeouts}, {9 Objectives, No Timeouts}, {9 Objectives, Timeouts}}

\nextgroupplot[group/empty plot]
\end{groupplot}

\node at (group c2r1.south) [below, yshift=-0.5cm] {\ref{optimalEvalLegend}};

\end{tikzpicture}

\vspace{-5mm}

\caption{Performance of exact algorithm on \mbox{TPC-H}: Prohibitive computational cost due to high number of Pareto plans (timeout at 2 hours)}
\label{optimalEvalFIg}
\end{figure}
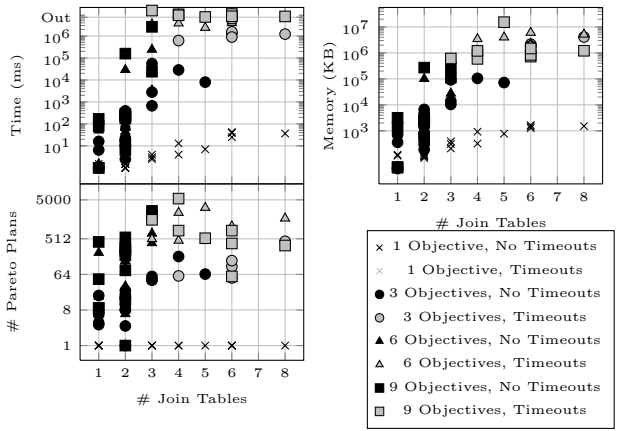

Optimizing for one objective never requires more than 100 milliseconds per query and never consumes more than 1.7~MB of main memory. For multiple objectives, the computational cost of the \AG{} becomes however quickly prohibitive with growing number of tables (referring to Question \textit{i)}). The \AG{} often reaches the timeout of two hours and allocates gigabytes of main memory during optimization. This happens already for queries joining only three tables; while the number of possible join orders is small in this case, the total search space size is already significant as over 10 different configurations are considered for the scan and for the join operator respectively (considering for instance different sample densities and different degrees of parallelism). 

Figure~\ref{optimalEvalFIg} explains the significant difference in time and space requirements between SOQO and MOQO: The number of Pareto plans per table set is always one for SOQO but grows quickly in the number of tables (and objectives) for MOQO. The space consumption of the \AG{} directly relates to the number of Pareto plans. The run time relates to the total number of considered plans which is much higher than the number of Pareto plans but directly correlated with it\footnote{All plans considered for joining a set of tables are combinations of two Pareto plans; the number of considered plans therefore grows quadratically in the number of Pareto plans.}. 
Discarding Pareto plans seems therefore the most natural way to increase efficiency (referring to Question~\textit{ii)}). 

Ganguly et al.~\cite{ganguly1992query} used an upper bound of $2^l$ ($l$ designates the number of objectives) on the number of Pareto plans per table set for their complexity analysis of the \AG{}. This bound derives from the optimistic assumption that different objectives are not correlated. Figure~\ref{optimalEvalFIg} shows that this bound is unrealistic (8, 64, and 512 are the theoretical bounds for 3, 6, and 9 objectives). The bound is a mismatch from the \textit{quantitative} perspective (as the bound is exceeded by orders of magnitude\footnote{We generate up to 443 Pareto plans on average when considering three objectives and up to 3157 plans when considering six objectives. }) and from the \textit{qualitative} perspective (as the number of Pareto plans seems to correlate with the search space size while the postulated bound only depends on the number of objectives). Therefore, this bound is not used in the following complexity analysis (referring to Question~\textit{iii)}). 

\subsection{Formal Complexity Analysis}

All query plans can be Pareto-optimal in the worst case (when considering at least two objectives). The following analysis remains unchanged under the assumption that a constant fraction of all possible plans is Pareto-optimal. If only one join operator is available, then the number of bushy plans for joining $n$ tables is given by $(2(n-1))!/(n-1)!$~\cite{ganguly1992query}. If $j$ scan and join operators are available, then the number of possible plans is given by 
\begin{equation*}
\mathcal{N}_{bushy}(j,n)=j^{2n-1}(2(n-1))!/(n-1)!.
\end{equation*}

\begin{theorem}
\label{optSpaceTheorem}
The \AG{} has space complexity
\begin{equation*}
O(\mathcal{N}_{bushy}(j,n)). 
\end{equation*}
\end{theorem}
\begin{proof}
Plan sets are the variables with dominant space requirements. A scan plan is represented by an operator ID and a table ID. All other plans are represented by the operator ID of the last join and pointers to the two sub-plans generating its operands. Therefore, each stored plan needs only $O(1)$ space. Each stored cost vector needs $O(1)$ space as well, since $l$ is a constant (see Section~\ref{modelSec}). 

Let $Q$ the set of tables to join. The \AG{} stores a set of Pareto plans for each non-empty subset of $Q$. The total number of stored plans is the sum of Pareto plans over all subsets. Let $k\in\{1,\ldots,|Q|\}$ and denote by $x_k$ the total number of Pareto plans, summing over all subsets of $Q$ with cardinality $k$. Each plan is Pareto-optimal in the worst case, therefore $x_k=\binom{n}{k}\mathcal{N}_{bushy}(j,k)$. It is $x_k\leq 2x_{k+1}$ for $k>1$. Therefore, the term $x_n=\mathcal{N}_{bushy}(j,n)$ dominates. The analysis is tight since the \AG{} has to store this number of plans in the worst case. 
\end{proof}

\begin{theorem}
The \AG{} has time complexity 
\begin{equation*}
O(\mathcal{N}^2_{bushy}(j,n)).
\end{equation*}
\end{theorem}
\begin{proof}
Every plan is compared with all other plans that generate the same result. So the time complexity grows quadratically in the number of Pareto plans and a similar reasoning as in the proof of Theorem~\ref{optSpaceTheorem} can be applied. 
\end{proof}

The main advantage of the Selinger algorithm for SOQO~\cite{Selinger1979} over a naive plan enumeration approach is that its complexity only depends on the number of table sets but not on the number of possible query plans. The preceding analysis shows that this advantage vanishes when generalizing the Selinger algorithm to multiple cost objectives (leading to the \AG{}). The complexity of the \AG{} is even worse than that of an approach that successively generates all possible plans while keeping only the best plan generated so far. 



%% file: Paper/sections/wMOQOsec.tex
\section{Approximating Weighted MOQO}
\label{wMOQOsec}

The \AG{} is computationally expensive since it generates all Pareto plans for each table set. We present a more efficient algorithm: the representative-tradeoffs algorithm (\AWM{}). The new algorithm generates an approximate Pareto plan set for each table set. The cardinality of the approximate Pareto set is much smaller than the cardinality of the Pareto set. Therefore, the \AWM{} has lower computational cost than the \AG{} while it formally guarantees to return a near-optimal plan. The \AWM{} exploits a property of the cost objectives that we call the \textit{principle of near-optimality}. We provide a formal definition in Section~\ref{ponoSub} and show that most relevant objectives in query optimization possess that property. We describe the \AWM{} in Section~\ref{wMOQOcode} and prove that it produces near-optimal plans. In Section~\ref{wMOQOanalysis}, we analyze its time and space complexity. We prove that its complexity is more similar to the complexity of SOQO algorithms than to the one of the \AG{}. 



\subsection{Principle of Near-Optimality}
\label{ponoSub}

The \textit{principle of optimality} states the following in the context of MOQO~\cite{ganguly1992query}: If the cost of the sub-plans within a query plan decreases, then the cost of the query plan cannot increase. A formal definition follows. 
\begin{definition}
\textbf{Principle of Optimality (POO). } Let $P$ a query plan with sub-plans $p_L$ and $p_R$. Derive $P^*$ from $P$ by replacing $p_L$ by $p_L^*$ and $p_R$ by $p_R^*$. Then $\mathbf{c}(p_L^*)\preceq\mathbf{c}(p_L)$ and $\mathbf{c}(p_R^*)\preceq\mathbf{c}(p_R)$ together imply $\mathbf{c}(P^*)\preceq\mathbf{c}(P)$. 
\end{definition}

The POO holds for all common cost objectives. The \AG{} generates optimal plans as long as the POO holds. We introduce a new property in analogy to the POO. The \textit{principle of near-optimality} intuitively states the following: If the cost of the sub-plans within a query plan increases by a certain percentage, then the cost of the query plan cannot increase by more than that percentage. 

\begin{definition}
\textbf{Principle of Near-Optimality (PONO). } Let $P$ a query plan with sub-plans $p_L$ and $p_R$ and pick an arbitrary $\alpha\geq1$. Derive $P^*$ from $P$ by replacing $p_L$ by $p_L^*$ and $p_R$ by $p_R^*$. Then $\mathbf{c}(p_L^*)\preceq_{\alpha}\mathbf{c}(p_L)$ and $\mathbf{c}(p_R^*)\preceq_{\alpha}\mathbf{c}(p_R)$ together imply $\mathbf{c}(P^*)\preceq_{\alpha}\mathbf{c}(P)$. 
\end{definition}

We will see that the PONO holds for the nine objectives described in Section~\ref{costSub} as well as for other common objectives. Cost formulas in QO are usually recursive and calculate the (estimated) cost of a plan out of the cost of its sub-plans. Different formulas apply for different objectives and for different operators. Most cost formulas only use the functions sum, maximum, minimum, and multiplication by a constant. The formula $\max(t_L,t_R)+t_M$ estimates for instance execution time of a plan whose final operation is a Sort-Merge join whose inputs are generated in parallel; the terms $t_L$ and $t_R$ represent the time for generating and sorting the left and right input operand and $t_M$ is the time for the final merge. Let $F$ any of the three binary functions sum, maximum, and minimum. Then $F(\alpha a,\alpha b)\leq \alpha F(a,b)$ for arbitrary positive operands $a,b$ and $\alpha\geq1$. Let $F(a)$ the function that multiplies its input by a constant. Then trivially $F(\alpha a)\leq\alpha F(a)$. Therefore, the PONO holds as long as cost formulas are combined out of the four aforementioned functions (this can be proven via structural induction). The formula for tuple loss is an exception since it multiplies two factors that depend on the tuple loss in the sub-plans: The tuple loss of a plan is estimated out of the tuple loss values $a$ and $b$ of its sub plans according to the formula $F(a,b)=1-(1-a)(1-b)$. It is $F(\alpha a,\alpha b)=\alpha(a+b)-\alpha^2ab$. This term is upper-bounded by $\alpha(a+b-ab)=\alpha F(a,b)$ since $0\leq a,b\leq 1$ and $\alpha\geq1$. Note that \textbf{failure probability} is calculated according to the same formula as tuple loss (if the probabilities that single operations fail are modeled as independent Bernoulli variables). Accumulative cost objectives such as \textbf{monetary cost} are calculated according to similar formulas as energy consumption. 

\subsection{Pseudo-Code and Near-Optimality Proof}
\label{wMOQOcode}

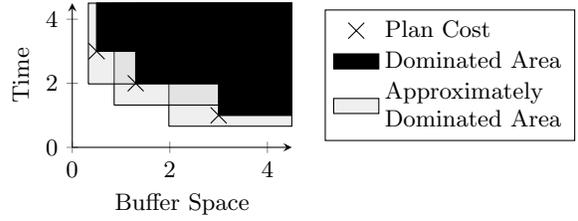
\begin{figure}
\input{figures/diagrams/approxPareto1.tex}

\vspace{-5mm}

\caption{Dominated versus approximately dominated area (with $\alpha=1.5$) in cost space\label{approximateDominanceFig}}
\centering
\end{figure}

We exploit the PONO to transform the \AG{} into an approximation scheme for \PW{}. Algorithm~\ref{weightedMOQOAlg} shows the parts of Algorithm~\ref{exactAlg} that need to be changed. The \AWM{} is the resulting approximation scheme. The \AWM{} takes a user-defined precision parameter $\alpha_U$ as input. It generates a plan whose weighted cost is not higher than the optimum by more than factor $\alpha_U$. We formally prove this statement later. The \AWM{} uses a different pruning function than the \AG{}: New plans are still compared with all plans that generate the same result. But new plans are only inserted if no other plan \textit{approximately} dominates the new one. This means that the \AWM{} tends to insert less plans than the \AG{}. Figure~\ref{approximateDominanceFig} helps to illustrate this statement: The \AG{} inserts new plans if their cost vector does not fall within the dominated area, the \AWM{} inserts new plans if their cost vector does neither fall into the dominated nor into the approximately dominated area. The following theorems exploit the PONO to show that the \AWM{} guarantees to generate near-optimal plans. They will implicitly justify the choice of the internal precision that is used during pruning. 

\begin{algorithm}[t!]
\renewcommand{\algorithmiccomment}[1]{// #1}
\begin{algorithmic}[1]
\State \Comment{Find $\alpha_U$-approximate plan for query $Q$, weights $\mathbf{W}$}
\Function{RTA}{$Q,\mathbf{W},\alpha_U$}
\State \Comment{Find $\alpha_U$-approximate Pareto plan set}
\State $\mathcal{P}\leftarrow\mathrm{FindParetoPlans}(Q,\alpha_U)$
\State \Comment{Return best plan in $\mathcal{P}$ for infinite bounds}
\State \Return SelectBest$(\mathcal{P},\mathbf{W},\mathbf{\infty})$
\EndFunction
\vspace{0.25cm}
\State \Comment{Find $\alpha_U$-approximate Pareto plan set}
\Function{FindParetoPlans}{$Q,\alpha_U$}
\Statex $\quad$ \Comment{Derive internal precision from $\alpha_U$}
\Statex $\quad$ $\alpha_i\leftarrow\sqrt[|Q|]{\alpha_U}$
\Statex $\quad$  ...
\State[13] \Comment{Prune access paths for single tables}
\Statex $\quad$ $\mathrm{Prune}(\mathcal{P}^q,p_N,\alpha_i)$
\Statex $\quad$  ...
\State[25] \Comment{Prune plans for non-singleton table sets}
\Statex $\quad$ $\mathrm{Prune}(\mathcal{P}^q,p_N,\alpha_i)$
\Statex $\quad$  ...
\EndFunction
\State \Comment{Prune set $\mathcal{P}$ with plan $p_N$ using precision $\alpha_i$}
\Procedure{Prune}{$\mathcal{P},p_N,\alpha_i$}
\State \Comment{Check whether new plan useful}
\If{$\neg\exists p\in\mathcal{P}:\mathbf{c}(p)\preceq_{\alpha_i}\mathbf{c}(p_N)$}
\Statex $\quad$ ...
\EndIf
\EndProcedure
\end{algorithmic}
\caption{The Representative-Tradeoffs Algorithm: An approximation scheme for Weighted MOQO. The code shows only the differences to Algorithm~\ref{exactAlg}.\label{weightedMOQOAlg}}
\end{algorithm}

\begin{theorem}
\label{approximateParetoSetTheorem}
The \AWM{} generates an $\alpha_i^{|Q|}$-approximate Pareto set. 
\end{theorem}
\begin{proof}
The proof uses induction over the number of tables $n=|Q|$. The \AWM{} examines all available access paths for single tables and generates an $\alpha_i$-approximate Pareto set. Assume \AWM{} generates $\alpha_i^{n}$-approximate Pareto sets for joining $n<N$ tables (inductional assumption). Let $p^*$ an arbitrary plan for joining $n=N$ tables and $p_L^*$, $p_R^*$ the two sub-plans generating the operands for the final join in $p^*$. Due to the inductional assumption, the \AWM{} generates a plan $p_L$ producing the same result as $p_L^*$ with $\mathbf{c}(p_L)\preceq_{\alpha_i^{N-1}}\mathbf{c}(p_L^*)$, and a plan $p_R$ producing the same result as $p_R^*$ with $\mathbf{c}(p_R)\preceq_{\alpha_i^{N-1}}\mathbf{c}(p_R^*)$. The plans $p_L$ and $p_R$ can be combined into a plan $p$ that generates the same result as $p^*$ and with $\mathbf{c}(p)\preceq_{\alpha_i^{N-1}}\mathbf{c}(p^*)$, due to the PONO. The \AWM{} might discard $p$ during the final pruning step but it keeps a plan $\widetilde{p}$ with $\mathbf{c}(\widetilde{p})\preceq_{\alpha_i}\mathbf{c}(p)$, therefore $\mathbf{c}(\widetilde{p})\preceq_{\alpha_i^{N}}\mathbf{c}(p^*)$ and the \AWM{} produces an $\alpha_i^N$-approximate Pareto set. 
\end{proof}

\begin{corollary}
\label{wMOQOprecisionTheorem}
The \AWM{} is an approximation scheme for weighted MOQO. 
\end{corollary}
\begin{proof}
The \AWM{} generates an $\alpha_U$-approximate Pareto set according to Theorem~\ref{approximateParetoSetTheorem} (since $\alpha_i^{|Q|}=\alpha_U$). This set contains a plan $p$ with $\mathbf{c}(p)\preceq_{\alpha_U}\mathbf{c}(p^*)$ for any optimal plan $p^*$. It is $C_{\mathbf{W}}(\mathbf{c}(p))\leq\alpha_U\cdot C_{\mathbf{W}}(\mathbf{c}(p^*))$ for arbitrary weights $\mathbf{W}$ and $p$ is therefore an $\alpha_U$-approximate solution. 
\end{proof}

The pruning procedure is sensitive to changes. It seems for instance tempting to reduce the number of stored plans further by discarding all plans that a newly inserted plan approximately dominates. Then the cost vectors of the stored plans can however depart more and more from the real Pareto frontier with every inserted plan. Therefore, the additional change would destroy near-optimality guarantees. 

\subsection{Complexity Analysis}
\label{wMOQOanalysis}

We analyze space and time complexity. The analysis is based on the following observations. 

\begin{observation}
\label{scanCostBound}
The cost of a plan that operates on a single table with $t$ tuples grows at most quadratically in $t$. 
\end{observation}

\begin{observation}
\label{joinCostBound}
Let $F(t_L,t_R,c_L,c_L)$ the recursive formula calculating---for a specific objective and operator---the cost of a plan whose final join has inputs with cardinalities $t_L$ and $t_R$ and generation costs $c_L$ and $c_R$. Then $F$ is in \begin{equation*}
O(t_Lc_R+c_L+(t_Lt_R)^2).
\end{equation*} 
\end{observation}

\begin{observation}
\label{lowerBoundAssumption}
There is an intrinsic constant for every objective such that the cost of all query plans for that objective is either zero or lower-bounded by that constant. 
\end{observation}

Observations~\ref{scanCostBound} and \ref{joinCostBound} trivially hold for objectives whose cost values are taken from an a-priori bounded domain such as reliability, coverage, or tuple loss (domain $[0,1]$). They clearly hold for objectives whose cost are proportional to input and output sizes\footnote{Using size and cardinality as synonyms is a simplification since tuple (byte) size may vary. It is however realistic to assume a constant upper bound for tuple sizes (e.g., the buffer page size). Also, the analysis can be generalized.} such as buffer or disc footprint (the maximal output cardinality of a join is $t_Lt_R$ which is dominated by the term $(t_Lt_R)^2$). Quicksort has quadratic worst-case complexity in the number of input tuples. It is the most expensive unary operation in our scenario, according to objectives such as time, energy, number of CPU cycles, and number of I/O operations. The (startup and total) time of a plan containing join operations can be decomposed into \textit{i)}~the time for generating the inputs to the final join, \textit{ii)}~the time for the join itself, \textit{iii)}~and the time for post-processing of the join result (e.g., sorting, hashing, materialization). The upper bound in Observation~\ref{joinCostBound} contains corresponding terms, taking into account that the right (inner) operand might have to be generated several times. It does not include terms representing costs for pre-processing join inputs (e.g., hashing) as this is counted as post-processing cost of the plan generating the corresponding operand. Observation~\ref{joinCostBound} can be justified similarly for objectives such as energy, number of CPU cycles, and number of I/O operations. 

Observation~\ref{lowerBoundAssumption} clearly holds for objectives with integer cost domains such as buffer and disc footprint (bytes), CPU cycles, time (in milliseconds), and number of used cores. It also covers objectives with non-discrete value domains such as tuple loss. Tuple loss has a non-discrete value domain since---given enough tables in which we can vary the sampling rate---the tuple loss values of different plans can get arbitrarily close to each other (e.g., compare tuple loss ratio of one plan sampling 1\% of every table with one that samples 2\% in one table and 1\% of the others, the values get closer the more tables we have). Assuming that the scan operators are parameterized by a discrete sampling rate (e.g., a percentage), there is still a gap between 0 and the minimal tuple loss ratio greater than zero. This gap does not depend on the number of tables (sampling at least one table with 99\% creates a tuple loss of at least 1\%). We derive a non-recursive upper bound on plan costs from our observations. 

\begin{lemma}
The cost of a plan joining $n$ tables of cardinality $m$ is bounded by $O(m^{2n})$ for every objective. 
\end{lemma}
\begin{proof}
Use induction over $n$. The lemma holds for $n=1$ due to Observation~\ref{scanCostBound}. Assume the lemma has been proven for $n<N$ (inductional assumption). Consider a join of $N$ tables. Cost is monotone in the number of processed tuples for any objective with non-bounded domain (not for tuple loss). So every join is a Cartesian product in the worst case and that implies $(t_Lt_R)^2=m^{2N}$. The inductional assumption implies $c_L+t_Lc_R\in O(m^{2N-1})$ so $(t_Lt_R)^2$ remains the dominant term. 
\end{proof}

The cost bounds allow to define an upper bound on the number of plans that the \AWM{} stores per table set. 

\begin{lemma}
\label{plansBoundLemma}
The \AWM{} stores $O((n\log_{\alpha_i}m)^{l-1})$ plans per table set. 
\end{lemma}
\begin{proof}
Function $\delta$ maps continuous cost vectors to discrete vectors such that $\delta^o(\mathbf{c})=\llcorner\log_{\alpha_i}(\mathbf{c}^o)\lrcorner$ for each objective $o$ and internal precision $\alpha_i$. If $\delta(\mathbf{c}_1)=\delta(\mathbf{c}_2)$ for two cost vectors $\mathbf{c}_1$ and $\mathbf{c}_2$, then $\mathbf{c}_1\preceq_{\alpha_i}\mathbf{c}_2$ and also $\mathbf{c}_2\preceq_{\alpha_i}\mathbf{c}_1$. This means that the cost vectors mutually approximately dominate each other. Therefore, the \AWM{} can never store two plans whose cost vectors are mapped to the same vector by $\delta$. The number of plans that have cost value zero for at least one objective is (asymptotically) dominated by the number of plans with non-zero cost values for every objective. Considering only the latter plans, their cost is lower-bounded by a constant (assume $1$ without restriction of generality) and upper-bounded by a function in $O(m^{2n})$. The cardinality of the image of $\delta$ is therefore upper-bounded by $O((n\log_{\alpha_i}m)^l)$. As the \AWM{} discards strictly dominated plans, the bound tightens to $O(l(n\log_{\alpha_i}m)^{l-1})$ which equals $O((n\log_{\alpha_i}m)^{l-1})$ since $l$ is constant (see Section~\ref{modelSec}).
\end{proof}

The function $\mathcal{N}_{stored}(m,n)=(n\log_{\alpha_i}m)^{l-1}$ denotes in the following the asymptotic bound on plan set cardinalities. 

\begin{theorem}
\label{wMOQOspaceTheorem}
The \AWM{} has space complexity
\begin{equation*}
O(2^n\mathcal{N}_{stored}(m,n)).
\end{equation*}
\end{theorem}
\begin{proof}
Plan sets are the variables with dominant space consumption in the \AWM{}. Each stored plan (with its associated cost vector) needs only $O(1)$ space as justified in the proof of Theorem~\ref{optSpaceTheorem}. Summing over all subsets of $Q$ yields the total complexity. 
\end{proof}



\begin{theorem}
The \AWM{} has time complexity 
\begin{equation*}
O(j3^{n}\mathcal{N}^{3}_{stored}(m,n)).
\end{equation*} 
\end{theorem}
\begin{proof}
There are $O(2^{k})$ possibilities of splitting a set of $k$ tables into two subsets. Every split allows to construct $O(j\mathcal{N}^2_{stored}(m,k-1))$ plans. Each newly generated plan is compared against all $O(\mathcal{N}_{stored}(m,k-1))$ plans in the set. Summing time complexity over all table sets yields $\sum_{k=1..n}\binom{n}{k}2^{k}j\mathcal{N}^{3}_{stored}(m,k)\leq j3^{n}\mathcal{N}_{stored}^{3}(m,n)$.
\end{proof}

Space and time complexity are exponential in the number of tables $n$. This cannot be avoided unless $P=NP$ since finding near-optimal query plans is already NP-hard for the single-objective case~\cite{Chatterji2002}. It is however remarkable that the time complexity of the \AWM{} differs only by factor $\mathcal{N}^3_{stored}(m,n)$ from the single-objective Selinger algorithm for bushy plans~\cite{Vance1996} (which has complexity $O(j3^n)$). This factor is a polynomial in number of join tables and table cardinalities. Unlike the \AG{}, the complexity of the \AWM{} does not depend on the total number of possible plans. This lets expect significantly better scalability (see Figure~\ref{complexityFig} for a visual comparison). Note that this qualitative result does not change when using different upper bounds on the recursive plan cost formulas (see Observation~\ref{scanCostBound} and \ref{joinCostBound}), as long as they remain polynomials. 

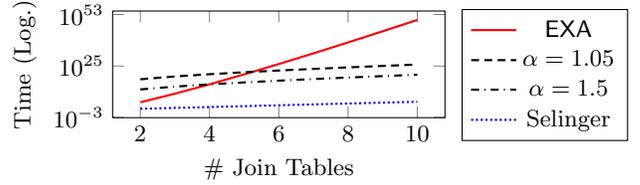
\begin{figure}
\begin{tikzpicture}
\begin{axis}[xlabel=\# Join Tables, ylabel=Time (Log.), xlabel near ticks, ylabel near ticks, height=3cm, width=6cm, ymode=log,
legend columns=1,legend style={at={(1.05,0.4)},anchor=west}]

\addplot[domain=2:10,samples at={2,3,4,5,6,7,8,9,10},
color=red, thick]{(6^(x-1+x)*factorial(2*(x-1))/factorial(x-1))^2};
\addplot[domain=2:10,samples at={2,3,4,5,6,7,8,9,10},
densely dashed, thick]{6*3^x*(x*ln(100000)/ln(1.05))^6};
\addplot[domain=2:10,samples at={2,3,4,5,6,7,8,9,10},
dashdotted, thick]{6*3^x*(x*ln(100000)/ln(1.5))^6};
\addplot[domain=2:10,samples at={2,3,4,5,6,7,8,9,10},
densely dotted, thick, color=blue]{6*3^x};

\addlegendentry{\AG{}}
\addlegendentry{$\alpha=1.05$}
\addlegendentry{$\alpha=1.5$}
\addlegendentry{Selinger}
\end{axis}
\end{tikzpicture}

\vspace{-5mm}

\caption{Comparing time complexity of the exact MOQO algorithm (\AG{}), the MOQO approximation scheme with $\alpha=1.05$ and $\alpha=1.5$, and Selinger's SOQO algorithm (Setting $j=6$, $l=3$, and $m=10^5$)\label{complexityFig}}
\end{figure}

%% file: figures/diagrams/approxPareto1.tex
\pgfmathparse{0.66}
\let\alphaC=\pgfmathresult

\pgfplotstableread{figures/diagrams/pareto.txt}\pareto

\begin{tikzpicture}

\begin{axis}[
xlabel=\reMetricOne{},
ylabel=\reMetricTwo{},
ylabel near ticks,
xlabel near ticks,
axis x line=bottom,
axis y line=left,
xmin=0,
ymin=0,
mark=x,
legend columns=1,
legend style={
at={(1.15,0.5)},
anchor=west},
legend cell align=left]

\addplot[only marks,mark=x,fill=black] table[x={x},y={y}]{\pareto};

\foreach \i in {2,1,0} {

\pgfplotstablegetelem{\i}{x}\of{\pareto}
\let\xCoord=\pgfplotsretval

\pgfplotstablegetelem{\i}{y}\of{\pareto}
\let\yCoord=\pgfplotsretval

\pgfplotstablegetelem{\i}{x}\of{\pareto}
\pgfmathparse{\alphaC * \pgfplotsretval}
\let\xCoordCoars=\pgfmathresult

\pgfplotstablegetelem{\i}{y}\of{\pareto}
\pgfmathparse{\alphaC * \pgfplotsretval}
\let\yCoordCoars=\pgfmathresult

\addplot[const plot mark right,fill=black,no markers, area legend, fill opacity=1.0] coordinates {
(\xCoord,\yCoord)
(4.5,\yCoord)
(4.5,4.5)
(\xCoord,4.5)
(\xCoord,\yCoord)
};

\addplot[const plot mark right,fill=lightgray,no markers, area legend, fill opacity=0.25] coordinates {
(\xCoordCoars,\yCoordCoars)
(4.5,\yCoordCoars)
(4.5,4.5)
(\xCoordCoars,4.5)
(\xCoordCoars,\yCoordCoars)
};
} 

\foreach \i in {2,1,0} {

\pgfplotstablegetelem{\i}{x}\of{\pareto}
\let\xCoord=\pgfplotsretval

\pgfplotstablegetelem{\i}{y}\of{\pareto}
\let\yCoord=\pgfplotsretval

\pgfplotstablegetelem{\i}{x}\of{\pareto}
\pgfmathparse{\alphaC * \pgfplotsretval}
\let\xCoordCoars=\pgfmathresult

\pgfplotstablegetelem{\i}{y}\of{\pareto}
\pgfmathparse{\alphaC * \pgfplotsretval}
\let\yCoordCoars=\pgfmathresult

\addplot[const plot mark right,fill=black,no markers, area legend, fill opacity=1.0] coordinates {
(\xCoord,\yCoord)
(4.5,\yCoord)
(4.5,4.5)
(\xCoord,4.5)
(\xCoord,\yCoord)
};
} 

\addlegendentry[align=left]{\planCost{}}
\addlegendentry[align=left]{\dominated{}}
\addlegendentry[align=left]{\approxDominated{}}

\end{axis}

\end{tikzpicture}

%% file: Paper/sections/bMOQOsec.tex
\section{Approximating Bounded MOQO}
\label{bMOQOsec}

\begin{figure}
\input{figures/diagrams/precisionProblem1.tex}

\vspace{-5mm}

\caption{An approximate Pareto set does not necessarily contain a near-optimal plan if bounds are considered}
\label{precisionProblemFig}
\end{figure}
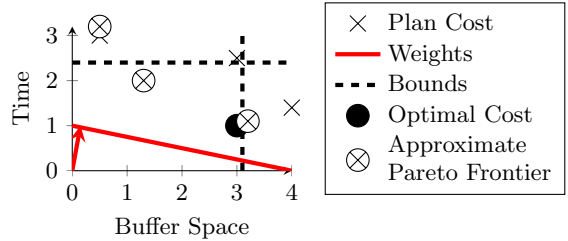

The \AWM{} selects an $\alpha_U$-approximate plan out of an $\alpha_U$-approximate Pareto set. This is always possible since similar cost vectors have similar weighted cost. This principle breaks when considering bounds in addition to weights. Even if two cost vectors are extremely similar, one of them can exceed the bounds while the other one does not. Figure~\ref{precisionProblemFig} illustrates this problem. There is no $\alpha\leq\alpha_U$ except $\alpha=1$ that guarantees a-priori that an $\alpha$-approximate Pareto set contains an $\alpha_U$-approximate plan. Choosing $\alpha=1$ leads however to high computational cost and should be avoided (the \AWM{} corresponds to the \AG{} if $\alpha=1$). 

Assuming that the pathological case depicted in Figure~\ref{precisionProblemFig} occurs always for $\alpha>1$ is however overly pessimistic. An $\alpha_U$-approximate Pareto set \textit{may} very well contain an $\alpha_U$-approximate solution. We present an iterative algorithm that exploits this fact: The iterative-refinement algorithm (\ABM{}) generates an approximate Pareto set in every iteration, starting optimistically with a coarse-grained approximation precision and refining the precision until a near-optimal plan is generated. This requires a stopping condition that detects whether an approximate Pareto set contains a near-optimal plan (without knowing the optimal plan or its cost). We present the \ABM{} and a corresponding stopping condition in Section~\ref{bAlgSub}. A potential drawback of an iterative approach is redundant work in different iterations. We analyze the complexity of the \ABM{} in Section~\ref{bAnaSub} and show how a carefully selected precision refinement policy makes sure that the amount of redundant work is negligible. We also prove that the \ABM{} always terminates. 

\subsection{Pseudo-Code and Near-Optimality Proof}
\label{bAlgSub}

\begin{algorithm}[t!]
\renewcommand{\algorithmiccomment}[1]{// #1}
\begin{algorithmic}[1]
\State \Comment{Find $\alpha_U$-approximate plan for query $Q$,}
\State \Comment{weights $\mathbf{W}$, bounds $\mathbf{B}$}
\Function{IRA}{$Q,\mathbf{W},\mathbf{B},\alpha_U$}
\State $i\leftarrow0$ \Comment{Initialize iteration counter}
\Repeat
\State $i\leftarrow i+1$ 
\State \Comment{Choose precision for this iteration}
\State $\alpha\leftarrow\alpha_U^{\displaystyle 2^{\displaystyle -i/(3l-3)}}$
\State \Comment{Find $\alpha$-approximate Pareto plan set}
\State $\mathcal{P}\leftarrow\mathrm{FindParetoPlans}(Q,\alpha)$
\State \Comment{Select best plan in $\mathcal{P}$}
\State $p_{opt}\leftarrow$SelectBest$(\mathcal{P},\mathbf{W},\mathbf{B})$
\Until{$\nexists p\in\mathcal{P}:\mathbf{c}(p)\preceq\alpha\mathbf{B}\wedge \frac{C_{\mathbf{W}}(\mathbf{c}(p))}{\alpha}<\frac{C_{\mathbf{W}}(\mathbf{c}(p_{opt}))}{\alpha_U}$}
\State \Return $p_{opt}$
\EndFunction
\end{algorithmic}
\caption{The Iterative-Refinement Algorithm: An Approximation Scheme for Bounded-Weighted MOQO. The Code Uses Sub-Functions From Algorithm~\ref{weightedMOQOAlg}.\label{boundedMOQOAlg}}
\end{algorithm}

Algorithm~\ref{boundedMOQOAlg} shows pseudo-code of the \ABM{}. The \ABM{} uses the functions {\codeF FindParetoPlans} and {\codeF SelectBest} which were already defined in Algorithm~\ref{weightedMOQOAlg}. The \ABM{} chooses in every iteration an approximation precision $\alpha$ and calculates an $\alpha$-approximate Pareto set. The precision gets refined from one iteration to the next. We will discuss the particular choice of precision formula in the next subsection. At the end of every iteration, the \ABM{} selects the best plan {\codeF $p_{opt}$} in the current approximate Pareto set. It terminates, once that plan is guaranteed to be $\alpha_U$-optimal. The stopping condition of the \ABM{} compares $p_{opt}$ with the best plan that can be found if the bounds are slightly relaxed (i.e., multiplied by a factor). This termination condition makes sure that the \ABM{} does not terminate before it finds an $\alpha_U$-approximate plan. This implies that the \ABM{} is an approximation scheme. 

\begin{theorem}
The \ABM{} is an approximation scheme for \PB{}. 
\end{theorem}
\begin{proof}
Denote by $\mathcal{P}$ the set of plans generated in the last iteration, by $\alpha$ the precision used in the last iteration, and by $p_{opt}$ the best plan in $\mathcal{P}$. The termination condition was met in the last iteration so there is no plan $p\in\mathcal{P}$ respecting the relaxed bounds $\alpha\mathbf{B}$ with $C_{\mathbf{W}}(\mathbf{c}(p))/\alpha< C_{\mathbf{W}}(\mathbf{c}(p_{opt}))/\alpha_U$. Let $p^{*}$ be an optimal plan for the input query (not necessarily contained in $\mathcal{P}$). 
Assume first that $p^{*}$ respects the bounds $\mathbf{B}$. Plan set $\mathcal{P}$ contains a plan $p_R$ whose cost vector is similar to the one of $p^*$: $\mathbf{c}(p_R)\preceq_{\alpha}\mathbf{c}(p^{*})$. The weighted cost of $p_R$ is near-optimal: $C_\mathbf{W}(\mathbf{c}(p_R))\leq \alpha C_\mathbf{W}(\mathbf{c}(p^*))$. Plan $p_R$ can violate the bounds $\mathbf{B}$ by factor $\alpha$ but respects the relaxed bounds: $\mathbf{c}(p_R)\preceq\alpha\mathbf{B}$. Let $p$ be the best plan in $\mathcal{P}$ that respects the relaxed bounds $\alpha\mathbf{B}$, the weighted cost of $p$ is smaller or equal to the one of $p_R$. Therefore, $C_\mathbf{W}(\mathbf{c}(p))/\alpha$ is a lower bound on $C_\mathbf{W}(p^*)$. If the weighted cost of $p_{opt}$ is not higher than that by more than factor $\alpha_U$, then $p_{opt}$ is an $\alpha_U$-approximate solution. Assume now that $p^{*}$ does not respect the bounds $\mathbf{B}$. Then no possible plan respects the bounds and weighted cost is the only criterion. Since $\alpha\leq\alpha_U$, the set $\mathcal{P}$ must contain an $\alpha_U$-approximate solution ($p_{opt}$). 
\end{proof}

\subsection{Analysis of Refinement Policy}
\label{bAnaSub}

The formula for calculating the approximation precision $\alpha$ should satisfy several requirements. First, the formula needs to be strictly monotonically decreasing in $i$ (the number of iterations) since the \ABM{} otherwise executes unnecessary iterations that do not generate new plans. Second, it should decrease quickly enough in $i$ such that the time required by the new iteration is higher or at least comparable to the time required in all previous iterations\footnote{Memory space can be reused in the new iteration so we only consider run time in the choice of $\alpha$. }. This ensures that the amount of redundant work is small compared with the total amount of work, as the \ABM{} can generate the same plans in several iterations. Third, it should decrease as slowly as the other requirements allow; choosing a lower $\alpha$ than necessary should be avoided, since the complexity of the Pareto set approximation grows quickly in the inverse of $\alpha$. The formula $\alpha=\alpha_U^{2^{-i/(3l-3)}}$ is strictly monotonically decreasing in $i$. It also satisfies the second and third requirement as we see next. The following theorem concerns space and time complexity of the $i$-th iteration of the \ABM{}. The proof is analogous to the proofs in Section~\ref{wMOQOanalysis}. 

\begin{theorem}
\label{bMOQOcomplexityTheorem}
The $i$-th iteration of the \ABM{} has

$\begin{array}{lr}
\text{space complexity} & O(2^n 2^{i/3}(n^2\log m /\log \alpha_U)^{l-1}), \\
\text{and time complexity} & O(j3^{n}2^i (n^2\log m/\log\alpha_U)^{3l-3}). 
\end{array}$
\end{theorem}

Assume that the time per iteration is \textit{proportional} to the worst-case complexity, or within a factor that does not depend on $i$ (but possibly on $n$, $m$, or $l$). Then the required time doubles from one iteration to the next, so the time of the last iteration is dominant. So the precision formula satisfies the second requirement and (approximately) the third, since decreasing iteration precision significantly slower would violate the second requirement. 

\begin{theorem}
The \ABM{} always terminates. 
\end{theorem}
\begin{proof}
For a fixed \PB{} instance $I=\langle Q,\mathbf{W},\mathbf{B}\rangle$ and plan space, there is only a finite number of possible query plans. Therefore, there is an $\alpha>1$ such that no plan $p$ exists which satisfies $\mathbf{c}(p)\preceq\alpha\mathbf{B}$ but not $\mathbf{c}(p)\preceq\mathbf{B}$. The precision refinement formula is strictly monotonically decreasing in $i$ (iteration counter). So the aforementioned $\alpha$ is reached after a finite number of iterations. Then the best plan that respects the strict bounds is equivalent to the best plan that respects the relaxed bounds, so the termination condition is satisfied. 
\end{proof}

%% file: figures/diagrams/precisionProblem1.tex
\pgfmathparse{1.25}
\let\alpha=\pgfmathresult

\pgfplotstableread{figures/diagrams/points.txt}\points
\pgfplotstableread{figures/diagrams/pareto.txt}\pareto
\pgfplotstableread{figures/diagrams/fine.txt}\fine
\pgfplotstableread{figures/diagrams/coarse.txt}\coarse

\begin{tikzpicture}

\begin{axis}[
xlabel=\reMetricOne{},
ylabel=\reMetricTwo{},
ylabel near ticks,
xlabel near ticks,
axis x line=bottom,
axis y line=left,
xmin=0,
ymin=0,
mark=x,
legend columns=1,
legend style={
at={(1.15,0.5)},
anchor=west},
legend cell align=left]

\addplot[only marks,mark=x] table[x={x},y={y}]{\points};
\addplot[only marks,forget plot] table[x={x},y={y}]{\coarse};
\addplot[only marks,forget plot] table[x={x},y={y}]{\fine};
\addplot[only marks,forget plot] table[x={x},y={y}]{\pareto};

\draw[red,-stealth,ultra thick]
(axis cs:0,0)
--
++
(axis direction cs:0.15,1);

\addplot[no marks,draw=red,ultra thick] coordinates {
(0,1)
(4,0)
};

\pgfmathparse{3.1}
\let\boundX=\pgfmathresult
\pgfmathparse{2.4}
\let\boundY=\pgfmathresult

\addplot[ultra thick,dashed,no marks] coordinates {
(\boundX,0)
(\boundX,3)
};
\addplot[ultra thick,dashed,no marks,forget plot] coordinates {
(0,\boundY)
(4,\boundY)
};

\addplot[mark=*,fill=black,only marks] coordinates {
(3,1)
};
\addplot[mark=otimes*,fill=white,only marks] coordinates {
(3.2,1.1)
(1.3,2)
(0.5,3.2)
};

\addlegendentry[align=left]{\planCost{}}
\addlegendentry[align=left]{\weights{}}
\addlegendentry[align=left]{\bounds{}}
\addlegendentry[align=left]{\optimalCost{}}
\addlegendentry[align=left]{Approximate\\ Pareto Frontier}

\end{axis}

\end{tikzpicture}

%% file: Paper/sections/experimentalSec.tex
\section{Experimental Evaluation}
\label{experimentalSec}

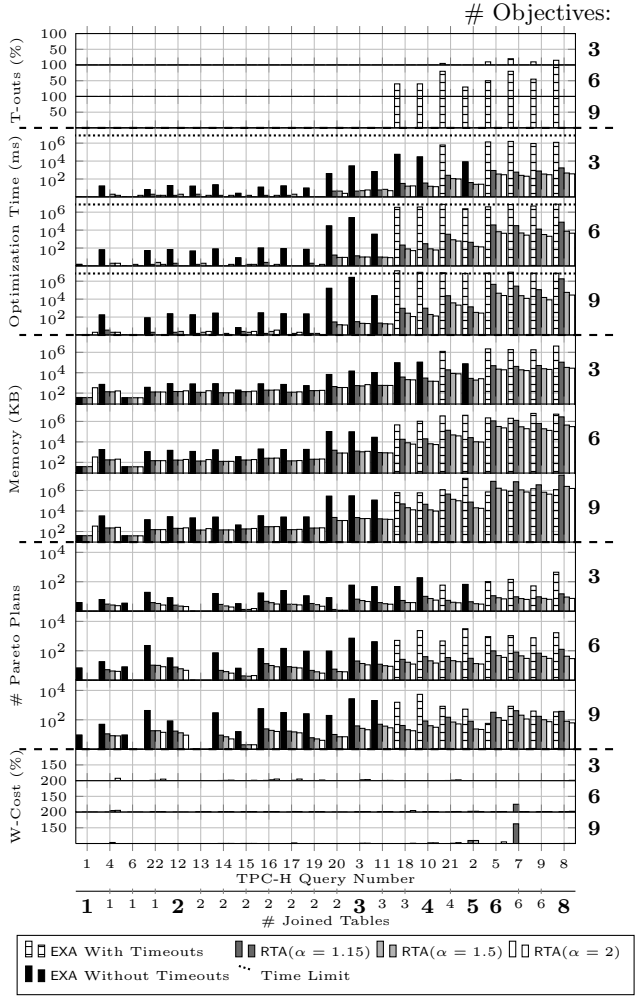
\begin{figure}[t!]
\centering

\begin{tikzpicture}
\begin{groupplot}[%
tiny,
width=\horScale*\benchWidth,height=2.5cm,
group style={ %
  group size=1 by 15, vertical sep=0cm,
  x descriptions at=edge bottom, 
  y descriptions at=edge left},
xmin=0.5, xmax=22.5, log origin y=0,
xlabel=TPC-H Query Number, xlabel shift={-0.15cm},
xtick={1, 2, 3, 4, 5, 6, 7, 8, 9, 10, 11, 12, 13, 14, 15, 16, 17, 18, 19, 20, 21, 22}, 
xticklabels={1, 4, 6, 22, 12, 13, 14, 15, 16, 17, 19, 20, 3, 11, 18, 10, 21, 2, 5, 7, 9, 8},
ymajorgrids, xmajorgrids, /pgf/bar width=2pt,
ybar=0pt,
cycle list name=benchmarkCycle]


\nextgroupplot[width=\horScale*\benchWidth,height=2cm, ymin=0, ymax=100,
ytick={50,100}]
\addplot table[header=false, x index=0, y index=\firstConfIndex] {benchmark/results/\threeObjWeighted/timeouts.txt};
\addplot table[header=false, x index=0, y index=\secondConfIndex] {benchmark/results/\threeObjWeighted/timeouts.txt};
\addplot table[header=false, x index=0, y index=\thirdConfIndex] {benchmark/results/\threeObjWeighted/timeouts.txt};
\addplot table[header=false, x index=0, y index=\fourthConfIndex] {benchmark/results/\threeObjWeighted/timeouts.txt};

\nextgroupplot[ylabel=T-outs (\%),
width=\horScale*\benchWidth,height=2cm, ymin=0, ymax=100,
ytick={50,100}]
\addplot table[header=false, x index=0, y index=\firstConfIndex] {benchmark/results/\sixObjWeighted/timeouts.txt};
\addplot table[header=false, x index=0, y index=\secondConfIndex] {benchmark/results/\sixObjWeighted/timeouts.txt};
\addplot table[header=false, x index=0, y index=\thirdConfIndex] {benchmark/results/\sixObjWeighted/timeouts.txt};
\addplot table[header=false, x index=0, y index=\fourthConfIndex] {benchmark/results/\sixObjWeighted/timeouts.txt};

\nextgroupplot[width=\horScale*\benchWidth,height=2cm, ymin=0, ymax=100,
ytick={50,100}]
\addplot table[header=false, x index=0, y index=\firstConfIndex] {benchmark/results/\nineObjWeighted/timeouts.txt};
\addplot table[header=false, x index=0, y index=\secondConfIndex] {benchmark/results/\nineObjWeighted/timeouts.txt};
\addplot table[header=false, x index=0, y index=\thirdConfIndex] {benchmark/results/\nineObjWeighted/timeouts.txt};
\addplot table[header=false, x index=0, y index=\fourthConfIndex] {benchmark/results/\nineObjWeighted/timeouts.txt};


\nextgroupplot[legend to name=wmoqoTimesLegend, legend style={anchor=north, legend columns=4, legend cell align=left},
ymode=log, ymin=1, ytickten={2, 4, 6}, 
ymax=50000000]
\addplot table[header=false, x index=0, y index=\firstConfIndex] {benchmark/results/\threeObjWeighted/time.txt};
\addplot table[header=false, x index=0, y index=\secondConfIndex] {benchmark/results/\threeObjWeighted/time.txt};
\addplot table[header=false, x index=0, y index=\thirdConfIndex] {benchmark/results/\threeObjWeighted/time.txt};
\addplot table[header=false, x index=0, y index=\fourthConfIndex] {benchmark/results/\threeObjWeighted/time.txt};

\addplot table[header=false, x index=0, y index=\firstConfIndex] {benchmark/results/\threeObjWeighted/timeNT.txt};
\addplot [sharp plot, thick, densely dotted] coordinates {(0,7200000) (23,7200000)};



\legend{
\AG{} With Timeouts,
\AWM{}($\alpha=1.15$),
\AWM{}($\alpha=1.5$),
\AWM{}($\alpha=2$), 
\AG{} Without Timeouts, Time Limit}

\nextgroupplot[ylabel=Optimization Time (ms),ymode=log,
ymin=1, ytickten={2, 4, 6}, 
ymax=50000000]
\addplot table[header=false, x index=0, y index=\firstConfIndex] {benchmark/results/\sixObjWeighted/time.txt};
\addplot table[header=false, x index=0, y index=\secondConfIndex] {benchmark/results/\sixObjWeighted/time.txt};
\addplot table[header=false, x index=0, y index=\thirdConfIndex] {benchmark/results/\sixObjWeighted/time.txt};
\addplot table[header=false, x index=0, y index=\fourthConfIndex] {benchmark/results/\sixObjWeighted/time.txt};

\addplot table[header=false, x index=0, y index=\firstConfIndex] {benchmark/results/\sixObjWeighted/timeNT.txt};
\addplot [sharp plot, thick, densely dotted] coordinates {(0,7200000) (23,7200000)};

\nextgroupplot[ymode=log,
ymin=1, ytickten={2, 4, 6}, 
ymax=50000000]
\addplot table[header=false, x index=0, y index=\firstConfIndex] {benchmark/results/\nineObjWeighted/time.txt};
\addplot table[header=false, x index=0, y index=\secondConfIndex] {benchmark/results/\nineObjWeighted/time.txt};
\addplot table[header=false, x index=0, y index=\thirdConfIndex] {benchmark/results/\nineObjWeighted/time.txt};
\addplot table[header=false, x index=0, y index=\fourthConfIndex] {benchmark/results/\nineObjWeighted/time.txt};

\addplot table[header=false, x index=0, y index=\firstConfIndex] {benchmark/results/\nineObjWeighted/timeNT.txt};
\addplot [sharp plot, thick, densely dotted] coordinates {(0,7200000) (23,7200000)};


\nextgroupplot[ymode=log,
ymax=50000000,ytickten={2,4,6}]
\addplot table[header=false, x index=0, y index=\firstConfIndex] {benchmark/results/\threeObjWeighted/memory.txt};
\addplot table[header=false, x index=0, y index=\secondConfIndex] {benchmark/results/\threeObjWeighted/memory.txt};
\addplot table[header=false, x index=0, y index=\thirdConfIndex] {benchmark/results/\threeObjWeighted/memory.txt};
\addplot table[header=false, x index=0, y index=\fourthConfIndex] {benchmark/results/\threeObjWeighted/memory.txt};
\addplot table[header=false, x index=0, y index=\firstConfIndex] {benchmark/results/\threeObjWeighted/memoryNT.txt};

\nextgroupplot[ylabel=Memory (KB),ymode=log,
ymax=50000000,ytickten={2,4,6}]
\addplot table[header=false, x index=0, y index=\firstConfIndex] {benchmark/results/\sixObjWeighted/memory.txt};
\addplot table[header=false, x index=0, y index=\secondConfIndex] {benchmark/results/\sixObjWeighted/memory.txt};
\addplot table[header=false, x index=0, y index=\thirdConfIndex] {benchmark/results/\sixObjWeighted/memory.txt};
\addplot table[header=false, x index=0, y index=\fourthConfIndex] {benchmark/results/\sixObjWeighted/memory.txt};
\addplot table[header=false, x index=0, y index=\firstConfIndex] {benchmark/results/\sixObjWeighted/memoryNT.txt};

\nextgroupplot[ymode=log,
ymax=50000000,ytickten={2,4,6}]
\addplot table[header=false, x index=0, y index=\firstConfIndex] {benchmark/results/\nineObjWeighted/memory.txt};
\addplot table[header=false, x index=0, y index=\secondConfIndex] {benchmark/results/\nineObjWeighted/memory.txt};
\addplot table[header=false, x index=0, y index=\thirdConfIndex] {benchmark/results/\nineObjWeighted/memory.txt};
\addplot table[header=false, x index=0, y index=\fourthConfIndex] {benchmark/results/\nineObjWeighted/memory.txt};
\addplot table[header=false, x index=0, y index=\firstConfIndex] {benchmark/results/\nineObjWeighted/memoryNT.txt};


\nextgroupplot[width=\horScale*\benchWidth,height=2.5cm,ymode=log,
ymin=1,ytickten={2,4},ymax=50000]
\addplot table[header=false, x index=0, y index=\firstConfIndex] {benchmark/results/\threeObjWeighted/paths.txt};
\addplot table[header=false, x index=0, y index=\secondConfIndex] {benchmark/results/\threeObjWeighted/paths.txt};
\addplot table[header=false, x index=0, y index=\thirdConfIndex] {benchmark/results/\threeObjWeighted/paths.txt};
\addplot table[header=false, x index=0, y index=\fourthConfIndex] {benchmark/results/\threeObjWeighted/paths.txt};
\addplot table[header=false, x index=0, y index=\firstConfIndex] {benchmark/results/\threeObjWeighted/pathsNT.txt};

\nextgroupplot[ylabel=\# Pareto Plans, width=\horScale*\benchWidth,height=2.5cm,ymode=log,
ymin=1,ytickten={2,4},ymax=50000]
\addplot table[header=false, x index=0, y index=\firstConfIndex] {benchmark/results/\sixObjWeighted/paths.txt};
\addplot table[header=false, x index=0, y index=\secondConfIndex] {benchmark/results/\sixObjWeighted/paths.txt};
\addplot table[header=false, x index=0, y index=\thirdConfIndex] {benchmark/results/\sixObjWeighted/paths.txt};
\addplot table[header=false, x index=0, y index=\fourthConfIndex] {benchmark/results/\sixObjWeighted/paths.txt};
\addplot table[header=false, x index=0, y index=\firstConfIndex] {benchmark/results/\sixObjWeighted/pathsNT.txt};

\nextgroupplot[width=\horScale*\benchWidth,height=2.5cm,ymode=log,
ymin=1,ytickten={2,4},ymax=50000]
\addplot table[header=false, x index=0, y index=\firstConfIndex] {benchmark/results/\nineObjWeighted/paths.txt};
\addplot table[header=false, x index=0, y index=\secondConfIndex] {benchmark/results/\nineObjWeighted/paths.txt};
\addplot table[header=false, x index=0, y index=\thirdConfIndex] {benchmark/results/\nineObjWeighted/paths.txt};
\addplot table[header=false, x index=0, y index=\fourthConfIndex] {benchmark/results/\nineObjWeighted/paths.txt};
\addplot table[header=false, x index=0, y index=\firstConfIndex] {benchmark/results/\nineObjWeighted/pathsNT.txt};


\nextgroupplot[width=\horScale*\benchWidth,height=2cm,ymin=100,ymax=200,ymode=normal,
ytick={150}]
\addplot coordinates {(0,0)};
\addplot table[header=false, x index=0, y index=\secondConfIndex] {benchmark/results/\threeObjWeighted/cost.txt};
\addplot table[header=false, x index=0, y index=\thirdConfIndex] {benchmark/results/\threeObjWeighted/cost.txt};
\addplot table[header=false, x index=0, y index=\fourthConfIndex] {benchmark/results/\threeObjWeighted/cost.txt};

\nextgroupplot[ylabel=W-Cost (\%), width=\horScale*\benchWidth,height=2cm,ymin=100,ymax=200,ymode=normal,
ytick={150,200}]
\addplot coordinates {(0,0)};
\addplot table[header=false, x index=0, y index=\secondConfIndex] {benchmark/results/\sixObjWeighted/cost.txt};
\addplot table[header=false, x index=0, y index=\thirdConfIndex] {benchmark/results/\sixObjWeighted/cost.txt};
\addplot table[header=false, x index=0, y index=\fourthConfIndex] {benchmark/results/\sixObjWeighted/cost.txt};

\nextgroupplot[width=\horScale*\benchWidth,height=2cm,ymin=100,
ymax=200,ymode=normal,
ytick={150,200}]
\addplot coordinates {(0,0)};
\addplot table[header=false, x index=0, y index=\secondConfIndex] {benchmark/results/\nineObjWeighted/cost.txt};
\addplot table[header=false, x index=0, y index=\thirdConfIndex] {benchmark/results/\nineObjWeighted/cost.txt};
\addplot table[header=false, x index=0, y index=\fourthConfIndex] {benchmark/results/\nineObjWeighted/cost.txt};
\end{groupplot}

\draw[thick, dashed] (group c1r3.south east) + (0.5cm, 0) -- (group c1r3.south west) -- +(-0.75cm,0);
\draw[thick, dashed] (group c1r6.south east) + (0.5cm, 0) -- (group c1r6.south west) -- +(-0.75cm,0);
\draw[thick, dashed] (group c1r9.south east) + (0.5cm, 0) -- (group c1r9.south west) -- +(-0.75cm,0);
\draw[thick, dashed] (group c1r12.south east) + (0.5cm, 0) -- (group c1r12.south west) -- +(-0.75cm,0);

\node at (group c1r1.north east) [above, xshift=-0.5cm, yshift=0cm] {\# Objectives:};

\node at (group c1r1.east) [xshift=0.25cm] {\small{\textbf{3}}};
\node at (group c1r2.east) [xshift=0.25cm] {\small{\textbf{6}}};
\node at (group c1r3.east) [xshift=0.25cm] {\small{\textbf{9}}};

\node at (group c1r4.east) [xshift=0.25cm] {\small{\textbf{3}}};
\node at (group c1r5.east) [xshift=0.25cm] {\small{\textbf{6}}};
\node at (group c1r6.east) [xshift=0.25cm] {\small{\textbf{9}}};

\node at (group c1r7.east) [xshift=0.25cm] {\small{\textbf{3}}};
\node at (group c1r8.east) [xshift=0.25cm] {\small{\textbf{6}}};
\node at (group c1r9.east) [xshift=0.25cm] {\small{\textbf{9}}};

\node at (group c1r10.east) [xshift=0.25cm] {\small{\textbf{3}}};
\node at (group c1r11.east) [xshift=0.25cm] {\small{\textbf{6}}};
\node at (group c1r12.east) [xshift=0.25cm] {\small{\textbf{9}}};

\node at (group c1r13.east) [xshift=0.25cm] {\small{\textbf{3}}};
\node at (group c1r14.east) [xshift=0.25cm] {\small{\textbf{6}}};
\node at (group c1r15.east) [xshift=0.25cm] {\small{\textbf{9}}};

\node at (group c1r15.south) [below, yshift=-1.1cm] {\ref{wmoqoTimesLegend}};

\node (left) at (group c1r15.south west) [above=-0.75cm] {}; 
\node (right) at (group c1r15.south east) [above=-0.75cm] {};

\draw (left.center) -- (right.center);
\def\tableCountArray{{1, 1, 1, 1, 2, 2, 2, 2, 2, 2, 2, 2, 3, 3, 3, 4, 4, 5, 6, 6, 6, 8}}
\foreach \x in {1, 2, 3, 5, 6, 7, 8, 9, 10, 11, 13, 14, 16, 19, 20}{
  \draw [gray] ($(left.center)!1/22*\x+1/44!(right.center)$) ++ (0,1pt) -- + (0,-2pt);
  \node [font=\tiny] at ($(left.center)!1/22*\x+1/44!(right.center)$) [below=0.0cm] {\pgfmathparse{\tableCountArray[\x]}\pgfmathresult};
}

\foreach \x in {0,4,12,15,17,18,21}{
  \node [font=\small] at ($(left.center)!1/22*\x+1/44!(right.center)$) [below=0.0cm] {\pgfmathparse{\tableCountArray[\x]}\textbf{\pgfmathresult}};
}

\node[font=\tiny] at (group c1r15.south) [below=0.85cm] {\# Joined Tables};

\end{tikzpicture}



\vspace{-5mm}

\caption{Optimizer performance comparison for weighted MOQO using timeout of two hours}
\label{wMOQOresultsFig}
\end{figure}


\label{wmoqoExperimentsSub}

We experimentally compare the approximation schemes against the \AG{}. The algorithms were implemented within the system described in Section~\ref{implementationSec}. A timeout of two hours was specified, using the technique outlined in Section~\ref{sub:optimalExperiments}. The experiments were executed on the hardware platform described in Section~\ref{sub:optimalExperiments}. 
We generated 20 test cases for each TPC-H query and three, six, and nine objectives respectively. Every test case is characterized by a set of considered objectives (selected randomly out of the nine implemented objectives), by weights on the selected objectives (chosen randomly from $[0,1]$ with uniform distribution), and (only for bounded MOQO) by bounds on a subset of the selected objectives. Bounds for objectives with a-priori bounded value domain (e.g., tuple loss with domain $[0,1]$) are chosen with uniform distribution from that domain. Bounds for objectives with non-bounded value domains (e.g., time) are chosen by multiplying the minimal possible value for the given objective and query by a factor chosen from $[1,2]$ with uniform distribution. 
 
Figure~\ref{wMOQOresultsFig} compares the performance of the \AG{} and the \AWM{} with $\alpha\in\{1.15,1.5,2\}$. Figure~\ref{wMOQOresultsFig} shows for each TPC-H query and each number of objectives \textit{i)}~the percentage of test cases that resulted in a timeout, and arithmetic average values for the metrics \textit{ii)}~optimization time (in milliseconds), \textit{iii})~allocated memory during optimization (in kilobytes), \textit{iv)}~number of Pareto plans for the last table set that was treated completely (before a timeout or before optimization finished), and \textit{v)}~weighted cost of the generated plan (as percentage of the optimal value over the plans generated by all algorithms for the same test case). Queries are ordered on the x-axis according to the maximal number of tables joined in any of their from-clauses as this relates to the search space size (with the caveats mentioned in Section~\ref{sub:optimalExperiments}). The time limit is marked by a dotted line in the subfigure showing optimization times. The fill pattern of the bars representing results for the \AG{} varies depending on whether the \AG{} had at least one timeout for the corresponding query and the corresponding number of objectives (the \AWM{} did not incur any timeouts). If the \AG{} had timeouts, then the reported values for time and memory consumption are lower bounds on the corresponding values for a completed optimizer run. 

The search space size correlates with the number of tables to join, and the number of objectives influences how many plans can be pruned during optimization. Therefore, the  percentage of timeouts (for the \AG{}), the optimization time, and the memory consumption all tend to increase in the number of objectives and the number of joined tables, as long as no timeouts distort the results. The \AG{} occasionally has timeouts already when considering only three objectives. For nine objectives, the \AG{} is not able to solve a single test case within the limit of two hours for queries that join more than three tables. Choices related to join order, operator selection, table sample density, and parallelization create a search space of considerable size, even for only four join tables. We have seen in Section~\ref{optimalEvalSec} that exact optimization takes less than 0.1 seconds despite the size of the search space, as long as only one objective is considered. Considering multiple objectives makes exact pruning however ineffective and leads to the high computational overhead of the \AG{}. The \AWM{} is often several orders of magnitude faster than the \AG{}. For nine objectives, the \AWM{} with $\alpha=1.15$ generates for instance near-optimal plans for \mbox{TPC-H} query~2 within less than 1.5 seconds average time. The \AG{} reaches the timeout of two hours for all 20 test cases. Optimization time and memory footprint decrease with increasing $\alpha$. This lets us expect that the \AWM{} can in principle scale to any MOQO problem when using an appropriate value for $\alpha$. 

The average quality of the plans produced by the \AWM{} is often significantly better than the worst case guarantees. Even for $\alpha=2$, the \AWM{} generates plans with an average cost overhead of below 1\% (100 times better than the theoretical bound) for 19 out of the 22 \mbox{TPC-H} queries. The Postgres optimizer selects locally optimal plans for the subqueries within a query. We left this mechanism in place as justified in Section~\ref{planSpaceSub}, even if it weakens the formal approximation guarantees for queries that contain subqueries (\mbox{TPC-H} queries 2, 4, 7, 8, 9, 11, 13, 15, 16, 17, 18, 20, 21, 22). In practice, the approximation guarantees were only violated in one case (\mbox{TPC-H} query~7) and only for specific choices of $\alpha$ ($\alpha=1.15$). 

Figure~\ref{bMOQOresultsFig} shows the results for bounded MOQO. The \AG{} is compared against the \ABM{} (instead of the \AWM{}) since only the \ABM{} guarantees to generate query plans that respect all hard bounds if such plans exist. Optimization always considers all nine objectives while the number of bounds varies between three and nine. Figure~\ref{bMOQOresultsFig} reports the number of iterations (instead of the number of Pareto plans), the reported numbers for memory consumption refer to the memory reserved in the last iteration (memory that was allocated before can be reused). The performance of the \AG{} is insensitive to the number of bounds. The performance of the \ABM{} varies with the number of bounds. Most significantly, time and memory consumption during optimization tend to be higher when hard bounds are set in comparison to the case without bonds. This can be seen by comparing Figure~\ref{bMOQOresultsFig} with Figure~\ref{wMOQOresultsFig}, as the \ABM{} behaves exactly like the \AWM{} if no bounds are specified. The reason is that the \ABM{} may have to choose a much smaller internal approximation factor than the \AWM{}, in order to verify if the best generated query plan is near-optimal among all plans respecting the bounds. The performance gap between approximate and exact MOQO is still significant: Summing over all test cases for bounded MOQO, the \AG{} had 464 timeouts while each \ABM{} instance had at most 4 timeouts. The total optimization time was more than 1200 hours for the \AG{} and less than 15 hours for the \ABM{} with $\alpha=1.15$. 

In some cases, the number of iterations of the \ABM{} increases with the user-defined approximation factor. If hard bounds are present then the internal approximation precision that is required to guarantee a near-optimal solution does not necessarily correlate with the approximation precision chosen by the user. However, even if the number of iterations increases, the total optimization time is not influenced significantly (except for queries with very low total optimization time where overhead by repeated query preprocessing can become a non-negligible component of the optimization time). This indicates that the time required for the first iterations is negligible compared with the time required for the last ones, which was an important criterion in the selection of our precision refinement policy.


\label{bmoqoExperimentsSub}

\begin{figure}[t!]
\begin{tikzpicture}
\begin{groupplot}[%
tiny,
width=\horScale*\benchWidth,height=2.5cm,
group style={ %
  group size=1 by 15, vertical sep=0cm,
  x descriptions at=edge bottom, 
  y descriptions at=edge left},
xmin=0.5, xmax=22.5, log origin y=0,
xlabel=TPC-H Query Number, xlabel shift={-0.15cm},
xtick={1, 2, 3, 4, 5, 6, 7, 8, 9, 10, 11, 12, 13, 14, 15, 16, 17, 18, 19, 20, 21, 22}, 
xticklabels={1, 4, 6, 22, 12, 13, 14, 15, 16, 17, 19, 20, 3, 11, 18, 10, 21, 2, 5, 7, 9, 8},
ymajorgrids, xmajorgrids, /pgf/bar width=2pt,
ybar=0pt,
cycle list name=benchmarkCycle]


\nextgroupplot[width=\horScale*\benchWidth,height=2cm, ymin=0, ymax=100,
ytick={50,100}]
\addplot table[header=false, x index=0, y index=\firstConfIndex] {benchmark/results/\threeObjBounded/timeouts.txt};
\addplot table[header=false, x index=0, y index=\secondConfIndex] {benchmark/results/\threeObjBounded/timeouts.txt};
\addplot table[header=false, x index=0, y index=\thirdConfIndex] {benchmark/results/\threeObjBounded/timeouts.txt};
\addplot table[header=false, x index=0, y index=\fourthConfIndex] {benchmark/results/\threeObjBounded/timeouts.txt};

\nextgroupplot[ylabel=T-outs (\%),
width=\horScale*\benchWidth,height=2cm, ymin=0, ymax=100,
ytick={50,100}]
\addplot table[header=false, x index=0, y index=\firstConfIndex] {benchmark/results/\sixObjBounded/timeouts.txt};
\addplot table[header=false, x index=0, y index=\secondConfIndex] {benchmark/results/\sixObjBounded/timeouts.txt};
\addplot table[header=false, x index=0, y index=\thirdConfIndex] {benchmark/results/\sixObjBounded/timeouts.txt};
\addplot table[header=false, x index=0, y index=\fourthConfIndex] {benchmark/results/\sixObjBounded/timeouts.txt};

\nextgroupplot[width=\horScale*\benchWidth,height=2cm, ymin=0, ymax=100,
ytick={50,100}]
\addplot table[header=false, x index=0, y index=\firstConfIndex] {benchmark/results/\nineObjBounded/timeouts.txt};
\addplot table[header=false, x index=0, y index=\secondConfIndex] {benchmark/results/\nineObjBounded/timeouts.txt};
\addplot table[header=false, x index=0, y index=\thirdConfIndex] {benchmark/results/\nineObjBounded/timeouts.txt};
\addplot table[header=false, x index=0, y index=\fourthConfIndex] {benchmark/results/\nineObjBounded/timeouts.txt};


\nextgroupplot[legend to name=wmoqoTimesLegend, legend style={anchor=north, legend columns=4, legend cell align=left},
ymode=log, ymin=1, ytickten={2, 4, 6}, 
ymax=50000000]
\addplot table[header=false, x index=0, y index=\firstConfIndex] {benchmark/results/\threeObjBounded/time.txt};
\addplot table[header=false, x index=0, y index=\secondConfIndex] {benchmark/results/\threeObjBounded/time.txt};
\addplot table[header=false, x index=0, y index=\thirdConfIndex] {benchmark/results/\threeObjBounded/time.txt};
\addplot table[header=false, x index=0, y index=\fourthConfIndex] {benchmark/results/\threeObjBounded/time.txt};

\addplot table[header=false, x index=0, y index=\firstConfIndex] {benchmark/results/\threeObjBounded/timeNT.txt};
\addplot [sharp plot, thick, densely dotted] coordinates {(0,7200000) (23,7200000)};




\legend{
\AG{} With Timeouts,
\ABM{}($\alpha=1.15$),
\ABM{}($\alpha=1.5$),
\ABM{}($\alpha=2$), 
\AG{} Without Timeouts, Time Limit}

\nextgroupplot[ylabel=Optimization Time (ms),ymode=log,
ymin=1, ytickten={2, 4, 6}, 
ymax=50000000]
\addplot table[header=false, x index=0, y index=\firstConfIndex] {benchmark/results/\sixObjBounded/time.txt};
\addplot table[header=false, x index=0, y index=\secondConfIndex] {benchmark/results/\sixObjBounded/time.txt};
\addplot table[header=false, x index=0, y index=\thirdConfIndex] {benchmark/results/\sixObjBounded/time.txt};
\addplot table[header=false, x index=0, y index=\fourthConfIndex] {benchmark/results/\sixObjBounded/time.txt};

\addplot table[header=false, x index=0, y index=\firstConfIndex] {benchmark/results/\sixObjBounded/timeNT.txt};
\addplot [sharp plot, thick, densely dotted] coordinates {(0,7200000) (23,7200000)};


\nextgroupplot[ymode=log,
ymin=1, ytickten={2, 4, 6}, 
ymax=50000000]
\addplot table[header=false, x index=0, y index=\firstConfIndex] {benchmark/results/\nineObjBounded/time.txt};
\addplot table[header=false, x index=0, y index=\secondConfIndex] {benchmark/results/\nineObjBounded/time.txt};
\addplot table[header=false, x index=0, y index=\thirdConfIndex] {benchmark/results/\nineObjBounded/time.txt};
\addplot table[header=false, x index=0, y index=\fourthConfIndex] {benchmark/results/\nineObjBounded/time.txt};

\addplot table[header=false, x index=0, y index=\firstConfIndex] {benchmark/results/\nineObjBounded/timeNT.txt};
\addplot [sharp plot, thick, densely dotted] coordinates {(0,7200000) (23,7200000)};



\nextgroupplot[ymode=log,
ymax=50000000,ytickten={2,4,6}]
\addplot table[header=false, x index=0, y index=\firstConfIndex] {benchmark/results/\threeObjBounded/LImemory.txt};
\addplot table[header=false, x index=0, y index=\secondConfIndex] {benchmark/results/\threeObjBounded/LImemory.txt};
\addplot table[header=false, x index=0, y index=\thirdConfIndex] {benchmark/results/\threeObjBounded/LImemory.txt};
\addplot table[header=false, x index=0, y index=\fourthConfIndex] {benchmark/results/\threeObjBounded/LImemory.txt};
\addplot table[header=false, x index=0, y index=\firstConfIndex] {benchmark/results/\threeObjBounded/LImemoryNT.txt};

\nextgroupplot[ylabel=Memory (KB),ymode=log,
ymax=50000000,ytickten={2,4,6}]
\addplot table[header=false, x index=0, y index=\firstConfIndex] {benchmark/results/\sixObjBounded/LImemory.txt};
\addplot table[header=false, x index=0, y index=\secondConfIndex] {benchmark/results/\sixObjBounded/LImemory.txt};
\addplot table[header=false, x index=0, y index=\thirdConfIndex] {benchmark/results/\sixObjBounded/LImemory.txt};
\addplot table[header=false, x index=0, y index=\fourthConfIndex] {benchmark/results/\sixObjBounded/LImemory.txt};
\addplot table[header=false, x index=0, y index=\firstConfIndex] {benchmark/results/\sixObjBounded/LImemoryNT.txt};

\nextgroupplot[ymode=log,
ymax=50000000,ytickten={2,4,6}]
\addplot table[header=false, x index=0, y index=\firstConfIndex] {benchmark/results/\nineObjBounded/LImemory.txt};
\addplot table[header=false, x index=0, y index=\secondConfIndex] {benchmark/results/\nineObjBounded/LImemory.txt};
\addplot table[header=false, x index=0, y index=\thirdConfIndex] {benchmark/results/\nineObjBounded/LImemory.txt};
\addplot table[header=false, x index=0, y index=\fourthConfIndex] {benchmark/results/\nineObjBounded/LImemory.txt};
\addplot table[header=false, x index=0, y index=\firstConfIndex] {benchmark/results/\nineObjBounded/LImemoryNT.txt};


\nextgroupplot[width=\horScale*\benchWidth,height=2.25cm,ymode=normal,ymax=120]
\addplot table[header=false, x index=0, y index=\firstConfIndex] {benchmark/results/\threeObjBounded/iterations.txt};
\addplot table[header=false, x index=0, y index=\secondConfIndex] {benchmark/results/\threeObjBounded/iterations.txt};
\addplot table[header=false, x index=0, y index=\thirdConfIndex] {benchmark/results/\threeObjBounded/iterations.txt};
\addplot table[header=false, x index=0, y index=\fourthConfIndex] {benchmark/results/\threeObjBounded/iterations.txt};

\nextgroupplot[ylabel=\# Iterations, width=\horScale*\benchWidth,height=2.25cm,ymode=normal,ymax=120]
\addplot table[header=false, x index=0, y index=\firstConfIndex] {benchmark/results/\sixObjBounded/iterations.txt};
\addplot table[header=false, x index=0, y index=\secondConfIndex] {benchmark/results/\sixObjBounded/iterations.txt};
\addplot table[header=false, x index=0, y index=\thirdConfIndex] {benchmark/results/\sixObjBounded/iterations.txt};
\addplot table[header=false, x index=0, y index=\fourthConfIndex] {benchmark/results/\sixObjBounded/iterations.txt};

\nextgroupplot[width=\horScale*\benchWidth,height=2.25cm,ymode=normal,ymax=120]
\addplot table[header=false, x index=0, y index=\firstConfIndex] {benchmark/results/\nineObjBounded/iterations.txt};
\addplot table[header=false, x index=0, y index=\secondConfIndex] {benchmark/results/\nineObjBounded/iterations.txt};
\addplot table[header=false, x index=0, y index=\thirdConfIndex] {benchmark/results/\nineObjBounded/iterations.txt};
\addplot table[header=false, x index=0, y index=\fourthConfIndex] {benchmark/results/\nineObjBounded/iterations.txt};


\nextgroupplot[width=\horScale*\benchWidth,height=2cm,ymin=100,ymax=200,ymode=normal,
ytick={150}]
\addplot coordinates {(0,0)};
\addplot table[header=false, x index=0, y index=\secondConfIndex] {benchmark/results/\threeObjBounded/cost.txt};
\addplot table[header=false, x index=0, y index=\thirdConfIndex] {benchmark/results/\threeObjBounded/cost.txt};
\addplot table[header=false, x index=0, y index=\fourthConfIndex] {benchmark/results/\threeObjBounded/cost.txt};

\nextgroupplot[ylabel=W-Cost (\%), width=\horScale*\benchWidth,height=2cm,ymin=100,ymax=200,ymode=normal,
ytick={150,200}]
\addplot coordinates {(0,0)};
\addplot table[header=false, x index=0, y index=\secondConfIndex] {benchmark/results/\sixObjBounded/cost.txt};
\addplot table[header=false, x index=0, y index=\thirdConfIndex] {benchmark/results/\sixObjBounded/cost.txt};
\addplot table[header=false, x index=0, y index=\fourthConfIndex] {benchmark/results/\sixObjBounded/cost.txt};

\nextgroupplot[width=\horScale*\benchWidth,height=2cm,ymin=100,
ymax=200,ymode=normal,
ytick={150,200}]
\addplot coordinates {(0,0)};
\addplot table[header=false, x index=0, y index=\secondConfIndex] {benchmark/results/\nineObjBounded/cost.txt};
\addplot table[header=false, x index=0, y index=\thirdConfIndex] {benchmark/results/\nineObjBounded/cost.txt};
\addplot table[header=false, x index=0, y index=\fourthConfIndex] {benchmark/results/\nineObjBounded/cost.txt};
\end{groupplot}

\draw[thick, dashed] (group c1r3.south east) + (0.5cm, 0) -- (group c1r3.south west) -- +(-0.75cm,0);
\draw[thick, dashed] (group c1r6.south east) + (0.5cm, 0) -- (group c1r6.south west) -- +(-0.75cm,0);
\draw[thick, dashed] (group c1r9.south east) + (0.5cm, 0) -- (group c1r9.south west) -- +(-0.75cm,0);
\draw[thick, dashed] (group c1r12.south east) + (0.5cm, 0) -- (group c1r12.south west) -- +(-0.75cm,0);

\node at (group c1r1.north east) [above, xshift=-0.5cm, yshift=0cm] {\# Bounds:};

\node at (group c1r1.east) [xshift=0.25cm] {\small{\textbf{3}}};
\node at (group c1r2.east) [xshift=0.25cm] {\small{\textbf{6}}};
\node at (group c1r3.east) [xshift=0.25cm] {\small{\textbf{9}}};

\node at (group c1r4.east) [xshift=0.25cm] {\small{\textbf{3}}};
\node at (group c1r5.east) [xshift=0.25cm] {\small{\textbf{6}}};
\node at (group c1r6.east) [xshift=0.25cm] {\small{\textbf{9}}};

\node at (group c1r7.east) [xshift=0.25cm] {\small{\textbf{3}}};
\node at (group c1r8.east) [xshift=0.25cm] {\small{\textbf{6}}};
\node at (group c1r9.east) [xshift=0.25cm] {\small{\textbf{9}}};

\node at (group c1r10.east) [xshift=0.25cm] {\small{\textbf{3}}};
\node at (group c1r11.east) [xshift=0.25cm] {\small{\textbf{6}}};
\node at (group c1r12.east) [xshift=0.25cm] {\small{\textbf{9}}};

\node at (group c1r13.east) [xshift=0.25cm] {\small{\textbf{3}}};
\node at (group c1r14.east) [xshift=0.25cm] {\small{\textbf{6}}};
\node at (group c1r15.east) [xshift=0.25cm] {\small{\textbf{9}}};

\node at (group c1r15.south) [below, yshift=-1.1cm] {\ref{wmoqoTimesLegend}};

\node (left) at (group c1r15.south west) [above=-0.75cm] {}; 
\node (right) at (group c1r15.south east) [above=-0.75cm] {};

\draw (left.center) -- (right.center);
\def\tableCountArray{{1, 1, 1, 1, 2, 2, 2, 2, 2, 2, 2, 2, 3, 3, 3, 4, 4, 5, 6, 6, 6, 8}}
\foreach \x in {1, 2, 3, 5, 6, 7, 8, 9, 10, 11, 13, 14, 16, 19, 20}{
  \draw [gray] ($(left.center)!1/22*\x+1/44!(right.center)$) ++ (0,1pt) -- + (0,-2pt);
  \node [font=\tiny] at ($(left.center)!1/22*\x+1/44!(right.center)$) [below=0.0cm] {\pgfmathparse{\tableCountArray[\x]}\pgfmathresult};
}

\foreach \x in {0,4,12,15,17,18,21}{
  \node [font=\small] at ($(left.center)!1/22*\x+1/44!(right.center)$) [below=0.0cm] {\pgfmathparse{\tableCountArray[\x]}\textbf{\pgfmathresult}};
}

\node[font=\tiny] at (group c1r15.south) [below=0.85cm] {\# Joined Tables};

\end{tikzpicture}



\vspace{-5mm}

\caption{Optimizer performance comparison for bounded MOQO using timeout of two hours}
\label{bMOQOresultsFig}
\end{figure}
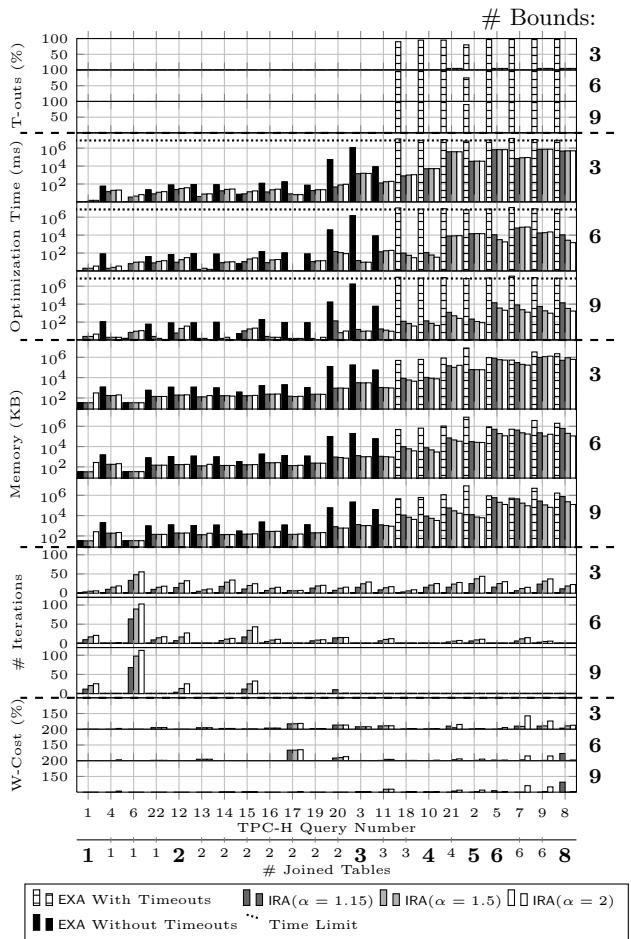

%% file: Paper/sections/conclusionSec.tex
Our MOQO approximation schemes find guaranteed near-optimal plans within seconds where exhaustive optimization takes hours. We analyzed the cost formulas of typical cost metrics in MOQO and identified common properties. We believe that our findings can be exploited for design and analysis of future MOQO algorithms. 

%% file: main.bbl
\begin{thebibliography}{10}

\medskip

\bibitem{abul2012alternative}
Z.~Abul-Basher, Y.~Feng, and P.~Godfrey.
\newblock {Alternative Query Optimization for Workload Management}.
\newblock In {\em Database and Expert Systems Applications}, 2012.

\bibitem{Agarwal2012}
S.~Agarwal, A.~Iyer, and A.~Panda.
\newblock {Blink and It's Done: Interactive Queries on Very Large Data}.
\newblock {\em PVLDB}, 2012.

\bibitem{Babcock2005}
B.~Babcock and S.~Chaudhuri.
\newblock {Towards a Robust Query Optimizer: a Principled and Practical
  Approach}.
\newblock {\em SIGMOD Conf.}, 2005.

\bibitem{Babu2005}
S.~Babu, P.~Bizarro, and D.~DeWitt.
\newblock {Proactive Re-Optimization}.
\newblock {\em SIGMOD Conf.}, 2005.

\bibitem{Bizarro2009}
P.~Bizarro, N.~Bruno, and D.~DeWitt.
\newblock {Progressive Parametric Query Optimization}.
\newblock {\em Trans. on KDE}, Apr. 2009.

\bibitem{Chatterji2002}
S.~Chatterji and S.~Evani.
\newblock {On the Complexity of Approximate Query Optimization}.
\newblock {\em PODS}, 2002.

\bibitem{Chu2002}
F.~Chu, J.~Halpern, and J.~Gehrke.
\newblock {Least Expected Cost Query Optimization: What can we Expect?}
\newblock {\em SIGMOD Conf.}, 2002.

\bibitem{Erlebach2014}
T.~Erlebach, H.~Kellerer, and U.~Pferschy.
\newblock {Approximating multiobjective knapsack problems}.
\newblock {\em Management Science}, pages 1603--1612, 2002.

\bibitem{Flach2010}
T.~Flach.
\newblock {Optimizing Query Execution to Improve the Energy Efficiency of
  Database Management Systems}.
\newblock Technical report, 2010.

\bibitem{Ganguly1998}
S.~Ganguly.
\newblock {Design and Analysis of Parametric Query Optimization Algorithms}.
\newblock {\em PVLDB}, 1998.

\bibitem{ganguly1992query}
S.~Ganguly, W.~Hasan, and R.~Krishnamurthy.
\newblock {Query Optimization for Parallel Execution}.
\newblock In {\em SIGMOD Conf.}, 1992.

\bibitem{garofalakis1996multi}
M.~Garofalakis and Y.~Ioannidis.
\newblock {Multi-Dimensional Resource Scheduling for Parallel Queries}.
\newblock In {\em SIGMOD Conf.}, 1996.

\bibitem{guha2003efficient}
S.~Guha, D.~Gunopoulos, N.~Koudas, D.~Srivastava, and M.~Vlachos.
\newblock {Efficient approximation of optimization queries under parametric
  aggregation constraints}.
\newblock In {\em PVLDB}, 2003.

\bibitem{Haas2003}
P.~Haas.
\newblock {Speeding Up DB2 UDB Using Sampling}.
\newblock {\em The IDUG Solutions Journal}, pages 1--10, 2003.

\bibitem{Hulgeri2004}
A.~Hulgeri.
\newblock {\em {Parametric Query Optimization}}.
\newblock PhD thesis, 2004.

\bibitem{Kambhampati02havasu:a}
S.~Kambhampati, U.~Nambiar, Z.~Nie, and S.~Vaddi.
\newblock {Havasu: A Multi-Objective, Adaptive Query Processing Framework for
  Web Data Integration}.
\newblock {\em ASU CSE}, 2002.

\bibitem{kllapi2011schedule}
H.~Kllapi, E.~Sitaridi, M.~M. Tsangaris, and Y.~E. Ioannidis.
\newblock {Schedule Optimization for Data Processing Flows on the Cloud}.
\newblock In {\em SIGMOD Conf.}, 2011.

\bibitem{kossmann2002shooting}
D.~Kossmann, F.~Ramsak, S.~Rost, and Others.
\newblock {Shooting Stars in the Sky: an Online Algorithm for Skyline Queries}.
\newblock In {\em PVLDB}, 2002.

\bibitem{Kossmann2000}
D.~Kossmann and K.~Stocker.
\newblock {Iterative Dynamic Programming: a New Class of Query Optimization
  Algorithms}.
\newblock {\em Trans. on Database Systems}, 1(212):43--82, 2000.

\bibitem{Marinescu2011}
R.~Marinescu.
\newblock {Efficient approximation algorithms for multi-objective constraint
  optimization}.
\newblock {\em Algorithmic Decision Theory}, 2011.

\bibitem{Papadimitriou2001}
C.~Papadimitriou and M.~Yannakakis.
\newblock {Multiobjective Query Optimization}.
\newblock {\em PODS}, 2001.

\bibitem{Selinger1979}
P.~G. Selinger, M.~M. Astrahan, D.~D. Chamberlin, R.~A. Lorie, and T.~G. Price.
\newblock {Access Path Selection in a Relational Database Management System}.
\newblock In {\em SIGMOD Conf.}, pages 23--34, 1979.

\bibitem{Simitsis2005}
A.~Simitsis, P.~Vassiliadis, and T.~Sellis.
\newblock {State-Space Optimization of ETL Workflows}.
\newblock {\em Trans. on KDE}, Oct. 2005.

\bibitem{Simitsis2012}
A.~Simitsis, K.~Wilkinson, M.~Castellanos, and U.~Dayal.
\newblock {Optimizing Analytic Data Flows for Multiple Execution Engines}.
\newblock {\em SIGMOD Conf.}, 2012.

\bibitem{TPC2013}
TPC.
\newblock {TPC-H Benchmark}, 2013.

\bibitem{Vance1996}
B.~Vance and D.~Maier.
\newblock {Rapid Bushy Join-Order Optimization with Cartesian Products}.
\newblock {\em SIGMOD Conf.}, 1996.

\bibitem{xu2012pet}
Z.~Xu, Y.~C. Tu, and X.~Wang.
\newblock {PET: Reducing Database Energy Cost via Query Optimization}.
\newblock {\em PVLDB}, 2012.

\end{thebibliography}
